\documentclass[11pt]{article}

\usepackage{mathrsfs}
\usepackage[deletedmarkup=xout]{changes}

\usepackage{amsthm,amsmath,amsfonts,amssymb,graphicx,subcaption}
\usepackage{authblk}
\usepackage{appendix}
\usepackage{bbm}
\usepackage{color}
\usepackage[font={small,rm}]{caption}

\font\tencmmib=cmmib10 \skewchar\tencmmib '60
\newfam\cmmibfam
\textfont\cmmibfam=\tencmmib

\def\lessim{\ \lower4pt\hbox{$
		\buildrel{\displaystyle <}\over\sim$}\ }
\def\gessim{\ \lower4pt\hbox{$\buildrel{\displaystyle >}
		\over\sim$}\ }

\def\la{\langle}
\def\ra{\rangle}

\def\TH{{\rm Th}}

\newcommand{\e}{\mathbb{E}}

\newtheorem{lemma}{\bf Lemma}
\newtheorem{definition}{\bf Definition}
\newtheorem{theorem}{\bf Theorem}

\newtheorem{remark}{\bf Remark}
\newtheorem{example}{\bf Example}

\newtheorem{proposition}{\bf Proposition}
\newtheorem{notation}{\bf Notation}
\newtheorem{basicsetting}{\bf Basic Setting}

\newenvironment{Proof of lemma}{\noindent{\bf Proof of Lemma}}{\hfill$\Box$\newline}
\newenvironment{Proof of theorem}{\noindent{\bf Proof of Theorem}}{\hfill{\footnotesize${\square}$}\newline}
\newenvironment{Proof of theorems}{\noindent{\bf Proof of Theorems}}{\hfill$\Box$\newline}
\newenvironment{Proof of proposition}{\noindent{\bf Proof of Proposition}}{\hfill$\Box$\newline}
\newenvironment{Proof of propositions}{\noindent{\bf Proof of Propositions}}{\hfill$\Box$\newline}
\newenvironment{Proof of exercise}{\noindent{\it Proof of Exercise:}}{\hfill$\Box$}

\definechangesauthor[color=orange]{ST}

\begin{document}
	
			\title{On convergence of the cavity and Bolthausen's \\
				TAP iterations to the local magnetization}
	
	\author{Wei-Kuo Chen\thanks{University of Minnesota. Email: wkchen@umn.edu.  Partly supported by NSF grant DMS-17-52184} \and  Si Tang \thanks{Lehigh University. Email: sit218@lehigh.edu. Partly supported by the Collaboration Grant from the Simons Foundation \#712728}}
		
	\maketitle
	
	\begin{abstract}
		The cavity and TAP equations are high-dimensional systems of nonlinear equations of the local magnetization in the Sherrington-Kirkpatrick model. In the seminal work \cite{Bolthausen14}, Bolthausen  introduced an iterative scheme that produces an asymptotic solution to the TAP equations if the model lies inside the Almeida-Thouless transition line. However, it was unclear if this asymptotic solution coincides with the local magnetization. In this work, motivated by the cavity equations, we introduce a new iterative scheme  and establish a weak law of large numbers. We show that our new scheme is asymptotically the same as the so-called Approximate Message Passing algorithm, a generalization of Bolthausen's iteration, that has been popularly adapted in compressed sensing, Bayesian inferences, etc. Based on this, we confirm that our cavity iteration and Bolthausen's scheme both converge to the local magnetization as long as the overlap is locally uniformly concentrated. 
	\end{abstract}


\section{Introduction}

For $n\geq 1$, denote by $[n]:=\{1,\ldots,n\}.$ Let $A_n=(a_{ij})_{i,j\in[n]}$ be a symmetric matrix satisfying that $a_{ii}=0$ for $i\in [n]$ and $a_{ij}$ are i.i.d. standard Gaussian random variables for $i< j.$ For a given (inverse) temperature $\beta>0$ and an external field $h>0$, define the Hamiltonian of the Sherrington-Kirkpatrick (SK) model as
\begin{align*}
H_{n,\beta,h}(\sigma)&=-\frac{\beta}{\sqrt{n}}\sum_{1\leq i<j\leq n}a_{ij}\sigma_i\sigma_j-h\sum_{i=1}^n\sigma_i
\end{align*}
for any $\sigma\in \{\pm 1\}^n$, and set the Gibbs measure on $\{\pm 1\}^n$ by
\begin{align*}
G_{n,\beta,h}(\sigma)=\frac{e^{-H_{n,\beta,h}(\sigma)}}{Z_{n,\beta,h}},
\end{align*}
where $Z_{n,\beta,h}$ is the normalizing constant, i.e., $Z_{n,\beta,h}:=\sum_{\sigma}e^{-H_{n,\beta,h}(\sigma)}.$ 
Denote by $\la \cdot\ra_{n,\beta,h}$ the expectation with respect to the Gibbs measure. Whenever there is no ambiguity, we will simply write $\la\cdot\ra_{n,\beta,h}$ by $\la \cdot\ra.$ 

The SK model  is a mean-field disordered spin system introduced in  \cite{SK72} to study some unusual magnetic behaviors of certain alloys. Although its formulation is very simple, the SK model exhibits very profound structures commonly shared in a number of disordered systems with large complexities. Using the replica method, the SK model has been intensively studied in the physics literature (see \cite{MPV87}). Rigorous mathematical treatments have also been successfully developed in the past decades (see \cite{Pan13,Tal111,Tal112}).

In this work, we investigate two classical approaches, the cavity method and the TAP equations, to studying the local magnetizations of spins $$\la \sigma\ra:=(\la \sigma_1\ra,\ldots,\la \sigma_n\ra)$$ in the SK model in the high-temperature regime.
Here, this  regime, denoted by $\mathcal{D}$, is defined as the collection of all pairs $\beta,h>0$ such that 
\begin{align}\label{add:overlap}
	\lim_{n\to\infty}\e\bigl\la \bigl|R(\sigma^1,\sigma^2)-q\bigr|^2\bigr\ra=0,
\end{align}
where  $R(\sigma^1,\sigma^2):=n^{-1}\sum_{i=1}^n\sigma_i^1\sigma_i^2$ is called the overlap of two  spin configurations $\sigma^1$ and $\sigma^2$ that are independently sampled from the Gibbs measure $G_{n,\beta,h}$.  The constant $q=q_{\beta,h}$ in \eqref{add:overlap} and hereafter is the unique solution to the following equation
$$
q_{\beta,h} =\e\tanh^2(\beta z\sqrt{q_{\beta,h}} +h)
$$
for any $\beta,h>0$ (see \cite{Chen17} and \cite[Proposition 1.3.8]{Tal111}).  Whenever \eqref{add:overlap} is satisfied, using the cavity method, Talagrand \cite[Proposition 1.6.8]{Tal111} showed that the limiting free energy is
\begin{align}\label{freeenergy}
	\lim_{n\to\infty}\frac{1}{n}\log Z_{n,\beta, h}=\log 2+\frac{\beta^2}{4}(1-q)^2+\e \log \cosh(\beta z\sqrt{q}+h)
\end{align}
for $z\thicksim N(0,1)$. In \cite{AT78}, de Almeida and Thouless conjectured that the high-temperature regime $\mathcal D$ can also be characterized by the so-called AT-line condition, that is, the collection $\mathcal A$ of all pairs $\beta, h>0$ such that
\begin{align}\label{ATline}
\beta^2\e\frac{1}{\cosh^4(\beta z\sqrt{q_{\beta,h}}+h)}\leq 1.
\end{align}
While it can be shown \cite{Chen17,JT17,Tal111,Ton02} that {$\mathcal{D}\subseteq \mathcal{A}$}, it was also understood in \cite{JT17,Tal112} that fairly large portions of {$\mathcal A$ is contained in $\mathcal D$.} However, a complete proof for $\mathcal{A}\subseteq \mathcal{D}$ remains missing. Incidentally, it was recently shown in \cite{WKC21} that if we replace the external field $h\sum_{i=1}^n\sigma_i$ by $\sum_{i=1}^n h_i\sigma_i$ for $h_1,\ldots,h_n$ i.i.d. centered normal, then the corresponding AT-line condition 
is indeed the  right curve to describe the high-temperature regime in the SK model. 

The asymptotic behavior of the local magnetizations can be described by the cavity equations and the TAP equations, both of which are high-dimensional systems of nonlinear equations.  
Initially proposed by M\'ezard-Parisi-Varosoro \cite{MPV87}, the cavity method allows one to compute asymptotically the local magnetization of an $n$-spin system through a nonlinear transformation of a Gaussian field in terms of the local magnetization of an $(n-1)$-spin system, namely,
\begin{align}\label{cavity}
	\la \sigma_n\ra \approx \tanh\Bigl(\frac{\beta}{\sqrt{n}}\sum_{j\neq n}a_{nj}\la \sigma_j\ra_{n-1,\beta',h}+h\Bigr),
\end{align}
where $\beta':=\beta\sqrt{(n-1)/n}.$ By symmetry, this equation is also valid for $\la \sigma_i\ra$, in which case, the local magnetizations on the right-hand side will correspond to the $(n-1)$-system excluding the $i$-th spin (see Lemma \ref{talagrand} below).

The TAP equations, named after Thouless, Anderson, and Palmer \cite{TAP}, describe the local magnetization from a different perspective. These equations assert that the local magnetization asymptotically satisfies a system of consistency equations,
\begin{align}\label{tap}
\la \sigma_i\ra&\approx \tanh\Bigl(\frac{\beta}{\sqrt{n}}\sum_{j\neq i}a_{ij}\la \sigma_j\ra+h-\beta^2\bigl(1-\bigl\|\la \sigma\ra\|^2\bigr)\la \sigma_i\ra\Bigr),\,\,\forall 1\leq i\leq n,
\end{align}
where 
$
\|x\|:=n^{-1}(\sum_{i=1}^n|x_i|^2)^{1/2}
$ for $x\in \mathbb{R}^n.$
 Here, the term $\beta^2\bigl(1-\bigl\|\la \sigma\ra\|^2\bigr)\la \sigma_i\ra$ (called the Onsager term) is introduced essentially to account for the substitution of $\la \sigma_j\ra_{n-1,\beta',h}$ in the cavity equations \eqref{cavity} by $\la \sigma_j\ra$, which is dependent on the entries $(a_{ij})_{j\neq i}$.

The systems of equations \eqref{cavity} and \eqref{tap} are valid for certain temperature $\beta$ and  external field $h.$ Assuming a very high temperature for the SK model, $\beta<1/2$, one can prove both the cavity equation and the TAP equations rigorously (see \cite{Chatterjee10}, \cite{Tal111}).   More subtle versions of the TAP equations in the entire temperature regime as well as for some variants of the SK model were also derived recently in \cite{AJ192,AJ191,BK19,CP2018,CPS18,CPS19}, where $\la \sigma\ra$ and the Onsager term were replaced by the notion of pure states or, more generally, the TAP states.

 It is natural to ask whether one can construct solutions to these equations asymptotically and show that they converge to the local magnetization in the entire high-temperature regime. The first attempt to this question was made by Bolthausen \cite{Bolthausen14}, in which he proposed an iterative scheme to construct an asymptotic solution to the TAP equations \eqref{tap}. More precisely, let $\mathbf{0}$ and $\mathbf{1}$ be the $n$-dimensional column vectors with all entries being $0$ and $1$, respectively. Starting from  $m^{[0]}=\mathbf{0}$ and $m^{[1]}=\sqrt{q_{\beta,h}}\mathbf{1}$, his iteration was defined as 
\begin{align*}
	m_i^{[k+1]}&=\tanh\Bigl(\frac{\beta}{\sqrt{n}}\sum_{j=1}^na_{ij}m_j^{[k]}+h-\beta^2\bigl(1-\|m^{[k]}\|^2\bigr)m_i^{[k-1]}\Bigr),\,1\leq i\leq n
\end{align*}
for $k\geq 1.$ Utilizing successive Gaussian conditioning arguments, it was shown in \cite{Bolthausen14} that this scheme converges in the sense that
\begin{align*}
	\lim_{k,k'\to\infty}\lim_{n\to\infty}\e\bigl\|m^{[k]}-m^{[k']}\bigr\|^2=0
\end{align*}
whenever $(\beta,h)$ lies in the regime $\mathcal{A}$,  but it was not answered whether his iteration converges to the local magnetization. 
In a more general formulation, Bolthausen's scheme is also known as the Approximate Message Passing (AMP) algorithm. Following the same conditioning argument in \cite{Bolthausen14}, one can show that this algorithm satisfies a law of large numbers, and efficient algorithms can be developed to solve many
estimation and optimization problems arising from compress sensing, Bayesian inference, etc.; see \cite{DJM13,DJM131,DMM09,DMM10,MV+18}. 

In this paper, motivated by the cavity equations, we propose a new nonlinear iterative scheme and establish three main results. First, we show that our scheme exhibits the same law of large numbers as the AMP algorithm. Second, we prove that our iteration based on the cavity equations produces asymptotically the same output as the AMP algorithm at all iterations. From these two results, we further establish that our and Bolthausen's iterations both converge to the local magnetization assuming that the overlap is locally uniformly concentrated. 

\section{Main results}

To prepare for the statements of our main results, we begin with

\begin{basicsetting}\label{basicsetting}\rm
	Let $u^n$ be an $n$-dimensional random vector independent of $A_n$ with $\|u^n\|\leq 1.$ Assume that the empirical distribution of $u^n$ converges to some {random variable} $W_0$ as $n\to\infty.$ As usual, we will simply write $u=u^n$ for notational clarity. Let $(f_k)_{k\geq 0}$ be a sequence of bounded and smooth functions on $\mathbb{R}$ with bounded derivatives of all orders. Whenever $f$ is a real-valued function on $\mathbb{R}$ and $w\in \mathbb{R}^n$, $f(w)\in \mathbb{R}^n$ is defined as a column vector $f(w)=(f(w_1),\dots, f(w_n))^T.$
\end{basicsetting}

\begin{definition}[Cavity Iteration]\label{myiteration}
		For each $n\geq 1$ and $0\leq k\leq n-1$, set
	\begin{align*}
		[n]_k=\bigl\{S\subseteq [n]\big||S|\leq n-(k+1)\bigr\}.
	\end{align*} 
	Let $n\geq 1.$ For any $S\in [n]_0$, define $w_S^{[0]}\in \mathbb{R}^{[n]\setminus S}$ by  $$\mbox{$w_{S,i}^{[0]}=u_i,\,\,\forall i\in [n]\setminus S.$}$$ For any $0\leq k\leq n-2$ and $S\in [n]_{k+1}$, define $w_S^{[k+1]}\in \mathbb{R}^{[n]\setminus S}$ iteratively by 
	\begin{align}\label{iteration}
		w_{S,i}^{[k+1]}=\frac{1}{\sqrt{n}}\sum_{j\notin S\cup \{i\}}a_{ij}f_k\bigl(w_{S\cup\{i\},j}^{[k]}\bigr),\,\,\forall i\in [n]\setminus S.
	\end{align}
	Finally, for $S=\emptyset$ and $0\leq k\leq n-1,$ we write $w^{[k]}= w_{\emptyset}^{[k]}\in \mathbb{R}^{[n]}$ and $w_{i}^{[k]}= w_{\emptyset, i}^{[k]}$ for each $i\in[n]$. 
\end{definition}

\begin{example}\rm The above definition gives that for $n\geq 2,$
	\begin{align*}
		w_i^{[1]}&=\frac{1}{\sqrt{n}}\sum_{j\neq i}a_{ij}f_0(u_j),\,\,i\in[n]
	\end{align*}
	and for $n\geq 3,$
	\begin{align*}
		w_i^{[2]}&=\frac{1}{\sqrt{n}}\sum_{j\neq i}a_{ij}f_1\bigl(w_{\{i\},j}^{[1]}\bigr)=\frac{1}{\sqrt{n}}\sum_{j\neq i}a_{ij}f_1\Bigl(\frac{1}{\sqrt{n}}\sum_{r\neq i,j}a_{jr}f_0(u_r)\Bigr),\,\,i\in[n].
	\end{align*}
	Also, for $n\geq 4,$
	\begin{align*}
		w_i^{[3]}&=\frac{1}{\sqrt{n}}\sum_{j\neq i}a_{ij}f_2\bigl(w_{\{i\},j}^{[2]}\bigr)\\
		&=\frac{1}{\sqrt{n}}\sum_{j\neq i}a_{ij}f_2\Bigl(\frac{1}{\sqrt{n}}\sum_{r\neq i,j}a_{jr}f_1\bigl(w_{\{i,j\},r}^{[1]}\bigr)\Bigr)\\
		&=\frac{1}{\sqrt{n}}\sum_{j\neq i}a_{ij}f_2\Bigl(\frac{1}{\sqrt{n}}\sum_{r\neq i,j}a_{jr}f_1\Bigl({\frac{1}{\sqrt{n}}\sum_{l\neq i,j,r}a_{rl}f_0(u_l)}\Bigr)\Bigr),\,\,i\in[n].
	\end{align*}
	We see that $w_{i}^{[3]}$ is implemented by considering all self-avoiding paths $i\to j\to r\to l$, as $j\neq i,$ $r\neq i,j,$ and $l\neq i,j,r.$ The computations of $w_{i}^{[1]}$, $w_{i}^{[2]}$ and $w_{i}^{[3]}$ essentially resemble that of $\la \sigma_n\ra_{n,\beta,h}$ by applying \eqref{talagrandlem:eq1} once, twice, and three times, respectively.
\end{example}
\begin{remark}\rm
Algorithms based on self-avoiding walks have been proposed in the literature, for example, in \cite{HS17} for community detection of sparse stochastic block model and in \cite{ding2020estimating} for the recovery problem in the generalized spiked Wigner model in the heavy-tailed setting. In these works, their iterations correspond to Definition \ref{myiteration} with the specific choice $f_{k}(x)=x$ for all $k\ge 0$. 
\end{remark}

In the iteration \eqref{iteration}, we exclude the columns and rows in $A_n$ corresponding to the set $S\cup\{i\}$ so that
$
\mbox{$(a_{ij})_{j\notin S\cup\{i\}}$ is independent of $\bigl(f_k(w_{S\cup\{i\},j}^{[k]})\bigr)_{j\notin S\cup\{i\}},$}
$
which readily implies that 	$w_{S,i}^{[k+1]}$ is a centered Gaussian random variable conditionally on 
$
\bigl(f_k(w_{S\cup\{i\},j}^{[k]})\bigr)_{j\notin S\cup\{i\}}.
$
Our first result establishes a weak law of large numbers for the random vectors $w^{[k]},w^{[k-1]},\ldots,w^{[0]}$.

\begin{theorem}\label{thm2}
	Let $k\geq 0.$ For any bounded Lipschitz function $\psi:\mathbb{R}^{k+1}\to \mathbb{R}$, we have that in probability,
	\begin{align*}
		\lim_{n\to\infty}\frac{1}{n}\sum_{i\in[n]}\psi\bigl(w_i^{[k]},w_i^{[k-1]},\ldots,w_i^{[0]}\bigr)&=\e\psi\bigl(W_{k},W_{k-1},\ldots,W_0\bigr),
	\end{align*}
	where $(W_k,\ldots, W_1)$ is jointly centered Gaussian independent of $W_0$ with covariance structure
	\begin{align}\label{lem2:eq3}
		\e W_{a+1}W_{b+1}=\e f_a(W_a)f_b(W_b)
	\end{align}
	for all $0\leq a,b\leq k-1.$
\end{theorem}

While the cavity iteration adapts self-avoiding paths, the AMP iteration is a mean-field method in the sense that all sites $i\in[n]$ are used without  preference.

\begin{definition}[AMP Iteration] Recall the $n$-dimensional random vector $u$ and the real-valued functions $(f_k)_{k\geq 0}$ considered in Basic Setting \ref{basicsetting}. Set $u^{[0]}=u$ and 
	\begin{align*}
		u_i^{[1]}&=\frac{1}{\sqrt{n}}\sum_{j=1}^na_{ij}f_0(u_j^{[0]}),\,\,\forall i\in [n].
	\end{align*}
	For $k\geq 1,$ the AMP iteration is defined as
	\begin{align}\label{amp}
		u_i^{[k+1]}&=\frac{1}{\sqrt{n}}\sum_{j=1}^na_{ij}f_k(u_j^{[k]})-\Bigl(\frac{1}{n}\sum_{j=1}^nf_k'(u_j^{[k]})\Bigr)f_{k-1}(u_i^{[k-1]}),\,\,\forall i\in [n].
	\end{align}
\end{definition}

As we have mentioned before, Bolthausen's iteration can be viewed as a special case of  the AMP algorithms. Specifically, it corresponds to the AMP iteration with $m^{[k]}=f_k(u^{[k]})$ and the following choice of functions, \begin{align}\label{add:eq21}
	\mbox{$u=\mathbf{0}$, $f_0(x)=0,$ $f_1(x)=\sqrt{q_{\beta,h}},$ and $f_k(x)=\tanh(\beta x+h)$ for all $k\geq 2$}.
\end{align}
Our next result shows that the iterative scheme  in Definition \ref{myiteration} is asymptotically the same as the AMP iteration.

\begin{theorem}\label{thm1}
	For any $k\geq 0$, there exists a constant $C_k>0$ such that for any $n\geq k+1$,
	\begin{align}\label{thm1:eq1}
	\e\bigl\|u^{[k]}-w^{[k]}\bigr\|^2\leq \frac{C_k}{n}.
	\end{align}
\end{theorem}

\begin{remark}
	\rm It was shown in \cite{DJM13} that the AMP iteration enjoys the same weak law of large numbers as Theorem~\ref{thm2}, where a Gaussian conditioning argument as in \cite{Bolthausen14} was adapted. Here, Theorems \ref{thm2} and \ref{thm1} together provide an independent proof for the convergence of the AMP iteration without using Gaussian conditioning.
\end{remark}

Our last result shows that Bolthausen's scheme converges to the local magnetization as long as the overlap is locally uniformly concentrated.

\begin{theorem}\label{thm0}
	Assume that $\beta,h>0$ satisfy that for some $\delta>0,$
	\begin{align}\label{shc}
		\lim_{n\to\infty}\sup_{\beta-\delta\leq \beta'\leq \beta}\e\bigl\la\bigl| R(\sigma^1,\sigma^2)-q_{\beta',h}\bigr|^2\bigr\ra_{n,\beta',h}=0.
	\end{align}
	We have that
	\begin{align*}
		\lim_{k\to\infty}\lim_{n\to\infty}\e\bigl\|\la \sigma\ra-m^{[k]}\bigr\|^2=0.
	\end{align*}
	In particular, here the inner limit exists for any $k\ge 0$.
\end{theorem}

The complexity of Bolthausen’s iteration is $O(n^2)$ and consequently, Theorem \ref{thm0} guarantees a polynomial-time  algorithm to approximate the local magnetization. Due to Theorem \ref{thm1}, our cavity iteration corresponding to \eqref{add:eq21} also converges to the local magnetization under the same assumption as Theorem \ref{thm0}.
In a related direction, we refer the readers to check \cite{Montanari18} for a polynomial-time algorithm to produce near-ground states in the SK
model via the AMP algorithm under the ``full replica symmetry breaking'' assumption. See more related results in \cite{EMS20,GJ2019,GJW2004,Subag2018}.

\begin{remark}\rm
The local magnetization is the barycenter of the Gibbs measure; when the high-temperature condition \eqref{add:overlap} is satisfied, for any $k\geq 2$ and i.i.d. samples $\sigma^1,\ldots,\sigma^k$ from the Gibbs measure, the vectors $\sigma^1-\la\sigma\ra,\ldots,\sigma^k-\la \sigma\ra$ are mutually orthogonal to each other and to the local magnetization. From these properties, it is tempting to believe that one can study the free energy of the SK model via large deviation techniques, by tilting the Gibbs measure according to $\la\sigma\ra$. This strategy was implemented  in \cite{Bolthausen19},  where the Gibbs measure was tilted with respect to $m^{[k]}$ at very high temperature. With the result of Theorem \ref{thm0}, it is of interest to see if one can establish the limiting free energy \eqref{freeenergy} of the SK model via large deviation arguments with respect to $\la \sigma\ra$.
\end{remark}

We close this section with a sketch of our proofs. Theorem \ref{thm2}  follows essentially from the way we define our scheme as its construction via self-avoiding paths already makes it clear on how we should manage the correlation between different layers. The proof of Theorem \ref{thm1} is the most delicate in this work; we have to remove all components corresponding to paths with loops in the AMP iteration $u^{[k+1]}$. While the basic idea is to rewrite $u_i^{[k+1]}$ by applying Taylor's theorem to the function $f_k$, the main challenge here is to carefully track the total error, again utilizing the self-avoiding feature of the paths along the iteration, see Section~\ref{sec:example} for an example and more detailed elaboration.  Finally, the proof of Theorem \ref{thm0} is based on the validities of Theorems \ref{thm2} and \ref{thm1}. We first argue that $m^{[k]}$ in Bolthausen's iteration is close to our scheme along with an explicit quantification of their distance, when  the high-temperature condition \eqref{shc} is in force. From this, Theorem~\ref{thm0} then follows immediately by the virtue of Theorem \ref{thm1}.
For the rest of the paper, Section~\ref{sec3} presents the proof of Theorem~\ref{thm0} assuming that Theorems \ref{thm2} and \ref{thm1} hold. Section \ref{sec4.1} establishes the weak law of large numbers of our scheme in Theorem~\ref{thm2}. Section \ref{sec4} prepares a number of moment controls for the partial derivatives of our scheme, which are the key ingredients in the proof of Theorem \ref{thm1} presented in Section \ref{sec5}.

\medskip
\medskip

{\noindent\bf Acknowledgements.} Both authors thank Antonio Auffinger for some useful discussions. 
In addition, they are grateful for the reviewer's careful reading and valuable comments regarding the presentation of this work.

\section{Proof of Theorem \ref{thm0}}\label{sec3}

In this section, we establish the proof of Theorem \ref{thm0} assuming the validity of Theorems \ref{thm2} and \ref{thm1}. First of all, we recall the statement of the cavity equations.

\begin{lemma}
	[Chapter 5 in \cite{MPV87} and Lemma 1.7.4 in \cite{Tal111}] \label{talagrandlem} If $\beta,h>0$ satisfy \eqref{shc}, then there exists a constant $\delta>0$ such that
	\begin{align}\label{talagrandlem:eq1}
		&\lim_{n\to\infty}\sup_{\beta-\delta\leq \beta'\leq \beta}\e\Bigl|\la \sigma_n\ra_{n,\beta',h}-\tanh\Bigl(\frac{\beta}{\sqrt{n}}\sum_{j\neq n}a_{nj}\la \sigma_j\ra_{n-1,\beta_n',h}+h\Bigr)\Bigr|^2=0
	\end{align}
	and
	\begin{align}\label{talagrandlem:eq2}
		&\lim_{n\to\infty}\sup_{\beta-\delta\leq \beta'\leq \beta}\e\bigl|\la \sigma_1\ra_{n,\beta',h}-\la \sigma_1\ra_{n-1,\beta_n',h}\bigr|^2=0,
	\end{align}
	where $\beta_n':=\beta' \sqrt{(n-1)/n}.$
\end{lemma}

\begin{remark}
	\rm  The original result in Talagrand's book \cite[Lemma 1.7.4]{Tal111} states only for  $\beta<1/2$ and $\delta=0$ instead of the locally uniformly limits. The condition $\beta<1/2$ ensures that  there exist some $K>0$ and $\delta>0$ such that $$\e\la \bigl|R(\sigma^1,\sigma^2)-q_{\beta,h}\bigr|^2\ra_{n,\beta,h}\leq \frac{K}{n}$$ for all $n\geq 1.$ Using this bound, his results stated that the expectations on the left-hand sides of  \eqref{talagrandlem:eq1} and \eqref{talagrandlem:eq2} are bounded above by $C/n$ for some universal constant $C>0$. If we now assume \eqref{shc} instead, the proof in  \cite[Lemma 1.7.4]{Tal111} still carries through for Lemma {\ref{talagrandlem}} without essential changes.
\end{remark}

We continue to restate Talagrand's lemma in a slightly more general formulation. Fix $\beta,h>0.$
Let $n\geq 2.$ For $S\subsetneq[n]$, consider the SK model on the sites $[n]\setminus S$ defined by
\begin{align*}
H_{S,n}(\sigma)&=-\frac{\beta}{\sqrt{n}}\sum_{i,j\in [n]\setminus S:i<j}a_{ij}\sigma_i\sigma_j-h\sum_{i\in [n]\setminus S}\sigma_i
\end{align*}
for all $\sigma\in \{\pm 1\}^{[n]\setminus S}.$ Note that when $S=\emptyset$, $H_{S,n}=H_n$.
Denote the Gibbs average associated to this Hamiltonian as $\la \cdot\ra_{n,\beta,h,S}$. Throughout the rest of the paper, for notational convenience, we denote this expectation simply by $\la \cdot\ra_S.$ We also set
$
\TH(x)=\tanh(x+h)
$ and denote $q=q_{\beta,h}.$ By the symmetry among sites, we can rewrite Lemma \ref{talagrandlem} as

\begin{lemma}\label{talagrand}
	Assume that $\beta,h>0$ satisfy \eqref{shc}. For any $k\geq 0$, we have that
	\begin{align}\label{talagrand:eq1}
	\lim_{n\to\infty}\sup_{(i,S):0\leq |S|\leq k,i\notin S}\e\Bigl|\la \sigma_i\ra_S-\TH\Bigl(\frac{\beta}{\sqrt{n}}\sum_{j\notin S\cup\{i\}}a_{ij}\la \sigma_j\ra_{S\cup\{i\}}\Bigr)\Bigr|^2=0
	\end{align}
	and
	\begin{align}\label{Talagrand:eq2}
	\lim_{n\to\infty}\sup_{(i,i',S):0\leq |S|\leq k,i,i'\notin S,i\neq i'}\e\bigl|\la \sigma_i\ra_S-\la \sigma_i\ra_{S\cup\{i'\}}\bigr|^2=0.
	\end{align}
\end{lemma}

\begin{proof}
Let $k\geq 0$ be fixed. Consider any $n>k.$ Let $(i,S)$ satisfy $S\subset [n]$ with $|S|\leq k$ and $i\notin S.$ Note that
\begin{align*}
	H_{S,n}(\sigma)&=-\frac{\beta'}{\sqrt{|[n]\setminus S|}}\sum_{s,t\in [n]\setminus S:s<t}a_{st}\sigma_s\sigma_t-h\sum_{s\in [n]\setminus S}\sigma_s
\end{align*} 
for $\sigma\in \{-1,1\}^{[n]\setminus S}$,
where $\beta':=\beta \sqrt{(n-|S|)/n}.$ In other words, $H_{S,n}(\sigma)$ can be regarded as the Hamiltonian of the SK model of size $n-|S|$ with temperature $\beta'$ and external field $h.$ Since $\beta(1-k/n)\leq \beta'\leq \beta$ and $\lim_{n\to\infty}\beta'=\beta$, our assertions follow from the symmetry among sites and Lemma \ref{talagrandlem}. 
\end{proof}

\subsection{Two crucial propositions}

We establish two important propositions in this subsection. First, we show that the summation in  \eqref{talagrand:eq1} can also be approximated by excluding one more row  and its corresponding column of the Gaussian matrix $(a_{r,r'})_{r,r'\in [n]\setminus(S\cup\{i\})}$ in $\la \sigma_j\ra_{S\cup\{i\}}.$ This will be used throughout the proof of Theorem~\ref{thm0}.

\begin{proposition}\label{add:lem1}
Assume that $\beta,h>0$ satisfy \eqref{shc}. For all $k\geq 2,$ we have that
\begin{align}\label{add:lem1:eq1}
\lim_{n\to\infty}\sup_{(i,i',S):0\leq |S|\leq k,i,i'\notin S,i\neq i'}\e\Bigl|\frac{1}{\sqrt{n}}\sum_{j\notin S\cup\{i\}}a_{ij}\la \sigma_j\ra_{S\cup\{i\}}-\frac{1}{\sqrt{n}}\sum_{j\notin S\cup\{i,i'\}}a_{ij}\la \sigma_j\ra_{S\cup\{i,i'\}}\Bigr|^2=0.
\end{align}
\end{proposition}
\begin{proof}
Note that the expectation in \eqref{add:lem1:eq1} is bounded from above by
\begin{align*}
&2\e\Bigl|\frac{1}{\sqrt{n}}\sum_{j\notin S\cup\{i,i'\}}a_{ij}\bigl(\la \sigma_j\ra_{S\cup\{i\}}-\la \sigma_j\ra_{S\cup\{i,i'\}}\bigr)\Bigr|^2+\frac{2}{n}\\
&=\frac{2}{n}\sum_{j\notin S\cup\{i,i'\}}\e\bigl|\la \sigma_j\ra_{S\cup\{i\}}-\la \sigma_j\ra_{S\cup\{i,i'\}}\bigr|^2+\frac{2}{n},
\end{align*}
where the equality here used the fact that $(a_{ij})_{j\notin S\cup\{ i,i'\}}$ is independent of $$
\bigl(\la \sigma_j\ra_{S\cup\{i\}}-\la \sigma_j\ra_{S\cup\{i,i'\}}\bigr)_{j\notin S\cup\{i,i'\}}.
$$ Using \eqref{Talagrand:eq2} completes our proof.
\end{proof}

Recall the iterative scheme $(w_{S}^{[k]})_{k\geq 0,S\subset[n]}$ from \eqref{iteration} with Basic Setting \ref{basicsetting}. The next proposition establishes an analogous statement as \eqref{Talagrand:eq2} for $w_S^{[k]},$ which will not only be critical to the proof of Theorem \ref{thm0}, but also to those of Theorems \ref{thm2} and \ref{thm1}. 

\begin{proposition}\label{lem3}
	For any $k\geq 0$ and $p\geq 1,$ there exists a constant $C_{k,p}>0$ such that for any $n\geq k+3,$
	\begin{align}\label{lem3:eq1}
	\sup\bigl(\e \bigl|w_{S,i}^{[k]}-w_{S\cup\{i'\},i}^{[k]}\bigr|^p\bigr)^{1/p}\leq \frac{C_{k,p}}{n^{1/2}},
	\end{align}
	where the supremum is over all $i,i'\in [n]$ and $S\subset [n]$ with $i\neq i'$, $i,i'\notin S,$ and $|S|\leq n-(k+2).$
\end{proposition}

\begin{proof}
It is easy to see that \eqref{lem3:eq1} is valid for $k=0$ and all $p\geq 1.$ Assume that \eqref{lem3:eq1} is valid for some $k\geq 0$ and all $p\geq 1.$ 
Consider an arbitrary $p\geq 1.$ Let $n\geq k+4.$ Fix $i,i'\in [n]$ and $S\subset [n]$ with $i\neq i'$, $i,i'\notin S,$ and $|S|\leq n-(k+3)$. Let $$
B_{l}:=f_{k}\bigl(w_{S\cup\{i\},l}^{[k]}\bigr)\,\,\mbox{and}\,\,D_l=f_{k}\bigl(w_{S\cup\{i,i'\},l}^{[k]}\bigr).$$
Observe that since the index $i$ does not appear in all indices of the Gaussian random variables in $(B_l)_{l\notin S\cup\{i,i'\}}$ and $(D_l)_{l\notin S\cup\{i,i'\}}$, we have that $(a_{il})_{l\notin S\cup\{i,i'\}}$ is independent of both $(B_l)_{l\notin S\cup\{i,i'\}}$ and $(D_l)_{l\notin S\cup\{i,i'\}}$. From this, we can write
\begin{align*}
w_{S,i}^{[k+1]}-w_{S\cup\{i'\},i}^{[k+1]}&=\frac{1}{\sqrt{n}}\sum_{l\notin S\cup\{i,i'\}}a_{il}(B_l-D_l)+\frac{1}{\sqrt{n}}a_{ii'}B_{i'}\\
&\stackrel{d}{=}z\Bigl(\frac{1}{n}\sum_{l\notin S\cup\{i,i'\}}(B_l-D_l)^2\Bigr)^{1/2}+\frac{1}{\sqrt{n}}a_{ii'}B_{i'},
\end{align*}
where $z$ is a standard normal random variable independent of $B_l$ and $D_l.$
Using the induction hypothesis and the fact that $f_k$'s are bounded and Lipschitz, it follows that
\begin{align*}
\bigl(\e\bigl|w_{S,i}^{[k+1]}-w_{S\cup\{i'\},i}^{[k+1]}\bigr|^p\bigr)^{1/p}&\leq \bigl(\e|z|^p\bigr)^{1/p}\Bigl(\frac{1}{n}\sum_{l\notin S\cup\{i,i'\}}\e |B_l-D_l|^{2p}\Bigr)^{1/2p}+\frac{\bigl(\e|z|^p\bigr)^{1/p}M_{k}}{n^{1/2}} \\
&\leq \frac{\bigl(\e|z|^p\bigr)^{1/p}C_{k,2p}}{n^{1/2}}+\frac{\bigl(\e|z|^p\bigr)^{1/p}M_{k}}{n^{1/2}},
\end{align*} 
where $M_{k}$ is the supremum norm of $f_{k}.$ This completes our proof.
\end{proof}

\subsection{Covariance structure}\label{sub3.2}

Recall $u$ and $(f_k)_{k\geq 0}$ from \eqref{add:eq21}. Recall the iterative scheme $w_{S}^{[k]}$ from \eqref{iteration} by applying the setting \eqref{add:eq21}. For $0\leq k\leq n-1$ and  any $S\in [n]_k$, set $\nu_{S}^{[k]}=\bigl(\nu_{S,i}^{[k]}\bigr)_{i\notin S}$ by
\begin{align*}
\nu _{S,i}^{[k]}=f_k\bigl(w_{S,i}^{[k]}\bigr),\,\,i\in [n]\setminus S.
\end{align*}
As before, if $S=\emptyset,$ we will simply denote $\nu_S^{[k]}$ by $\nu^{[k]}.$ Define the overlap between $\la \sigma\ra_S$ and $\nu_S^{[k]}$ by
\begin{align*}
R_S^k&=\frac{1}{n}\sum_{j\notin S}\la \sigma_j\ra_{S}\nu _{S,j}^{[k]}
\end{align*}
and denote
\begin{align*}
D_{S}=\frac{1}{n}\sum_{j\notin S}\la \sigma_j\ra_{S}^2,\quad
E_{S}^{k}=\frac{1}{n}\sum_{j\notin  S}\nu _{S,j}^{[k]2}.
\end{align*}
Define an auxiliary function $\Gamma(t;\gamma,\gamma')$ for $t\in [-1,1]$ and $\gamma,\gamma'\geq 0$ by
\begin{align*}
\Gamma(t;{\gamma,\gamma'})&:=\e \TH\bigl(\beta z\sqrt{\gamma|t|}+\beta z_1\sqrt{\gamma(1-|t|)}\bigr)\\
&\quad\cdot\TH\bigl(\beta\mbox{\rm sign}(t)z\sqrt{\gamma'|t|}+\beta z_2\sqrt{\gamma'(1-|t|)}\bigr)
\end{align*}
for $z,z_1,z_2$ i.i.d. standard Gaussian.  The following proposition takes care of the limits of $D_S,E_S^k,R_S^k.$

\begin{proposition}\label{lem-1}
	Assume that $\beta,h>0$ satisfy \eqref{shc}.	For any $k\geq 2$ and $\ell\geq 0,$ we have that
	\begin{align}
	\begin{split}\label{assertion:2}
	\lim_{n\to\infty}\sup_{|S|=\ell}\e\bigl|D_S-q\bigr|^2=0,\\
	\lim_{n\to\infty}\sup_{|S|=\ell}\e\bigl|E_S^k-q\bigr|^2=0.
	\end{split}
	\end{align}
   Furthermore,
	\begin{align}
	\begin{split}\label{assertion:1}
	\lim_{n\to\infty}\sup_{|S|=\ell}\e\Bigl|R_S^k-\Delta^{\circ(k-1)}\bigl(Q(\beta,h)\bigr)\Bigr|^2=0,
	\end{split}
	\end{align}
	where $Q(\beta,h):=\sqrt{q}\e \TH(\beta z\sqrt{q})$ and 
	\begin{align}\label{delta}
	\Delta\bigl(t\bigr)=\Gamma\bigl(t/q;q,q\bigr),\,\,t\in [-q,q].
	\end{align}
	The notation $\Delta^{\circ(k-1)}$ here means the composition of $\Delta$ for $(k-1)$ times.
\end{proposition}

For the rest of this subsection, we establish this proposition. 

\begin{notation}\label{notation1} \rm For two sequences of random variables $(a_n)_{n\geq 1}$ and $(b_n)_{n\geq 1}$, we say that $a_n\asymp_1 b_n$ if $\lim_{n\to\infty}\e|a_n-b_n|=0.$ It is straightforward that if $a_{n}\asymp_1 b_{n}$ and $c_{n}\asymp_1 d_{n}$ then (i) $f(a_{n})\asymp_{1} f(b_{n})$ for any Lipschitz function $f$ and (ii) $a_{n}c_{n}\asymp_1 b_{n}d_{n}$ provided $\sup_{n\ge 1}\{ |a_{n}|, |b_{n}|, |c_{n}|, |d_{n}|\} < \infty$. Also, for any $i\neq i',$ we use $\e_{i}$ and $\e_{i,i'}$ to denote the expectations with respect to $(a_{ij})_{j\in [n]}$ and $(a_{ij},a_{i'j})_{j\in[n]}$, respectively. 
	\end{notation}

\begin{proof}[\bf Proof of \eqref{assertion:2} in Proposition \ref{lem-1}:]
	Let $k\geq 2$ and $\ell\geq 0.$ Applying \eqref{Talagrand:eq2} and Proposition \ref{lem3} for $\ell$ many times, we have that uniformly over all $S$ with $|S|=\ell,$
	\begin{align*}
	D_S\asymp_1\frac{1}{n}\sum_{j=1}^n\la \sigma_j\ra^2\,\,\mbox{and}\,\,
    E_S^k\asymp_1\frac{1}{n}\sum_{j=1}^n\nu _{j}^{[k]2}.
	\end{align*}
	From \eqref{shc}, in probability,
	\begin{align*}
	\frac{1}{n}\sum_{j=1}^n\la \sigma_j\ra^2&=\bigl\la R(\sigma^1,\sigma^2)\bigr\ra\to q.
	\end{align*}
	Also, from Theorem \ref{thm2}, we see that $W_k\thicksim N(0,q)$ for $k\geq 2$ so that in probability,
	\begin{align*}
	\frac{1}{n}\sum_{j=1}^n\nu _{j}^{[k]2}&\to \e f_k^2\bigl(W_{k}\bigr)=\e\TH^2(\beta z \sqrt{q})=q.
	\end{align*}
	These imply the announced statement.
\end{proof}

The proof of \eqref{assertion:1} in Proposition \ref{lem-1} requires two lemmas. First, we show that the overlap $R_S^{k+1}$ satisfies the following recursive formula. Set
\begin{align*}
\rho_{S}^k&=\frac{R_S^k}{\sqrt{D_{S}E_{S}^k}}.
\end{align*}

\begin{lemma}\label{add:lem3}
	 Assume that $\beta,h>0$ satisfy \eqref{shc}. For any $k\geq 1$ and $\ell\geq 0,$ 
	\begin{align*}
		&\lim_{n\to\infty}\sup_{|S|=\ell}\e\Bigl|R_S^{k+1}-\frac{1}{n}\sum_{i\notin S}\Gamma\bigl(\rho_{S\cup\{i\}}^k;D_{S\cup\{i\}},E_{S\cup\{i\}}^{k}\bigr)\Bigr|^2=0.
		\end{align*}
\end{lemma}

\begin{proof} 
	Writing by using conditional expectations, 
	\begin{align}
		\begin{split}\label{add:eq22}
			&\e\Bigl| R_{S}^{k+1} - \frac{1}{n}\sum_{i\not \in S} \e_{i}\bigl[\la \sigma_i\ra_S \nu_{S,i}^{[k+1]}\bigr]\Bigr|^2\\
			&=\frac{1}{n^2}\sum_{i,i'\notin S:i\neq i'}\e\Bigl[\e_{i,i'} \bigl[\la \sigma_i\ra_S \nu_{S,i}^{[k+1]}\la \sigma_{i'}\ra_S \nu_{S,i'}^{[k+1]}\bigr]+\e_{i}\bigl[ \la \sigma_{i}\ra_S \nu_{S,i}^{[k+1]}\bigr]\cdot\e_{i'}\bigl[ \la \sigma_{i'}\ra_S \nu_{S,i'}^{[k+1]}\bigr]
			\\
			&\qquad\qquad-\la \sigma_i\ra_S \nu_{S,i}^{[k+1]}\cdot\e_{i'} \bigl[\la \sigma_{i'}\ra_S  \nu_{S,i'}^{[k+1]}\bigr]-\la \sigma_{i'}\ra_S \nu_{S,i'}^{[k+1]}\cdot\e_{i} \bigl[\la \sigma_{i}\ra_S  \nu_{S,i}^{[k+1]}\bigr]\Bigr]+O(n^{-1}),
		\end{split}
	\end{align}
	where $O(n^{-1})$ arises from the total contribution of the terms for $i=i' \in [n]$. To handle the terms inside the summations, note that from Lemma \ref{talagrand} and Propositions \ref{add:lem1} and \ref{lem3}, we have that 
  uniformly over all $(i,i',S)$ with $|S|=\ell$, $i,i'\notin S$, and $i\neq i',$ 
  \begin{align}\label{add:eq-1}
  \la \sigma_i\ra_S \nu_{S,i}^{[k+1]}&\asymp_1 \Theta_{S,i} \asymp_1 \Theta_{S,i,i'},
  \end{align}
where
  \begin{align*}
	\Theta_{S,i}&:=\TH\Bigl(\frac{\beta}{\sqrt{n}}\sum_{j\notin S\cup\{i\}}a_{ij}\la \sigma_j\ra_{S\cup\{i\}}\Bigr)
	\TH\Bigl(\frac{\beta}{\sqrt{n}}\sum_{j\notin S\cup\{i\}}a_{ij}\nu _{S\cup\{i\},j}^{[k]}\Bigr),\\
	\Theta_{S,i,i'}&:=\TH\Bigl(\frac{\beta}{\sqrt{n}}\sum_{j\notin S
		\cup\{i,i'\}}a_{ij}\la \sigma_j\ra_{S\cup\{i,i'\}}\Bigr)
	\TH\Bigl(\frac{\beta}{\sqrt{n}}\sum_{j\notin S\cup\{ i,i'\}}a_{ij}\nu _{S\cup\{i,i'\},j}^{[k]}\Bigr).
\end{align*}
   Here, note that $(a_{ij})_{j\notin S\cup\{i\}}$ is independent of $\la \sigma\ra_{S\cup\{i\}}$ and $\nu_{S\cup\{i\}}^{[k]}$ and that $(a_{ij})_{j\notin S\cup\{i,i'\}}$ is independent of $\la \sigma\ra_{S\cup\{i,i'\}}$ and $\nu _{S\cup\{i,i'\}}^{[k]}$. It follows that uniformly over all $(i,i',S)$ with $|S|=\ell$, $i,i'\notin S$, and $i\neq i',$
  \begin{align}
  \label{eqn:asymp-expectation}
    \e_{i}\bigl[\la \sigma_i\ra_S \nu_{S,i}^{[k+1]}\bigr] \asymp_{1}\e_i\bigl[\Theta_{S,i,i'}\bigr]\asymp_{1}\e_i\bigl[\Theta_{S,i}\bigr]=\Gamma\bigl(\rho_{S\cup\{i\}}^k;D_{S\cup\{i\}},E_{S\cup\{i\}}^{k}\bigr),
  \end{align}
which implies 
		\begin{align}\label{add:eq-2}
		&\lim_{n\to\infty}\sup_{|S|=\ell}\e\Bigl|\frac{1}{n}\sum_{i\not \in S} \e_{i}\bigl[\la \sigma_i\ra_S \nu_{S,i}^{[k+1]}\bigr] -\frac{1}{n}\sum_{i\notin S}\Gamma\bigl(\rho_{S\cup\{i\}}^k;D_{S\cup\{i\}},E_{S\cup\{i\}}^{k}\bigr)\Bigr|^2=0.
		\end{align}
In a similar manner, by \eqref{add:eq-1}, we have that uniformly over all $(i,i',S)$ with $|S|=\ell$, $i,i'\notin S,$ and $i\neq i',$
  \begin{align}
  \begin{split}\label{add:eq19}
  \e_{i,i'}\bigl[	\la \sigma_i\ra_S \nu_{S,i}^{[k+1]}\la \sigma_{i'}\ra_S \nu_{S,i'}^{[k+1]}\bigr]
  &\asymp_1 \e_{i,i'}\bigl[\Theta_{S,i,i'}\Theta_{S,i',i}\bigr] \asymp_1 \e_i\bigl[\Theta_{S,i,i'}\bigr]\e_{i'}\bigl[\Theta_{S,i',i}\bigr],
  \end{split}
  \end{align}
  where the second asymptotics is valid since $(a_{ij})_{j\notin S\cup\{i,i'\}}$ is independent of $(a_{i'j})_{j\notin S\cup\{i,i'\}}$.
  In addition, 
  \begin{align}
  \label{add:eq20}
  \e_{i,i'}\bigl[\la \sigma_i\ra_S \nu_{S,i}^{[k+1]}\cdot\e_{i'} \bigl[\la \sigma_{i'}\ra_S  \nu_{S,i'}^{[k+1]}\bigr]\bigr]
  &\asymp_1 \e_{i,i'}\bigl[\Theta_{S,i,i'}\e_{i'}[\Theta_{S,i',i}]\bigr]
  \asymp_1\e_i\bigl[\Theta_{S,i,i'}\bigr]\e_{i'}[\Theta_{S,i',i}].
  \end{align}
  Plugging  \eqref{eqn:asymp-expectation}, \eqref{add:eq19}, and \eqref{add:eq20} into \eqref{add:eq22}, we see that the right-hand side of \eqref{add:eq22} vanishes. Finally, applying \eqref{add:eq-2} to the left-hand side of \eqref{add:eq22} completes our proof.
\end{proof}

Next we show that the averaging local magnetization converges.

\begin{lemma}
	\label{add:lem2} Assume that $\beta,h>0$ satisfy \eqref{shc}. We have that in probability,
	\begin{align*}
	\lim_{n\to\infty}\frac{1}{n}\sum_{i\in [n]}\la \sigma_i\ra=\e \TH(\beta z\sqrt{q}).
	\end{align*}
\end{lemma}

\begin{proof}
	Let $\phi$ and $\psi$ be any two Lipschitz continuous functions on $[-1,1].$ From Lemma \ref{talagrand}, Propositions \ref{add:lem1}, and noting that for distinct $i,i'$, $(a_{ij})_{j\notin \{i,i'\}}$ and $(a_{i'j})_{j\notin \{i,i'\}}$ are independent each other, it follows that uniformly over any $i\neq i',$ 
	\begin{align*}
	\e_{i,i'}\phi(\la \sigma_i\ra)\psi(\la \sigma_{i'}\ra)&\asymp_1\e_i \phi\Bigl(\TH\Bigl(\frac{\beta}{\sqrt{n}}\sum_{j\notin \{i,i'\}}a_{ij}\la \sigma_j\ra_{\{i,i'\}}\Bigr)\Bigr)\cdot \e_{i'} \psi\Bigl(\TH\Bigl(\frac{\beta}{\sqrt{n}}\sum_{j\notin \{i,i'\}}a_{i'j}\la \sigma_j\ra_{\{i,i'\}}\Bigr)\Bigr)\\
	&=\e_z \phi\bigl(\TH\bigl(\beta z\sqrt{D_{\{i,i'\}}}\bigr)\bigr)\cdot \e_z \psi\bigl(\TH\bigl(\beta z\sqrt{D_{\{i,i'\}}}\bigr)\bigr),
	\end{align*}
	where the asymptotics are valid since $\phi,\psi,$ and $\TH$ are Lipschitz and $\e_z$ is the expectation with respect to $z$ only.  Next, from \eqref{shc} and \eqref{Talagrand:eq2}, $$
	q\asymp_1\la R(\sigma^1,\sigma^2)\ra=\frac{1}{n}\sum_{j\in [n]}\la \sigma_j\ra^2\asymp_1 \frac{1}{n}\sum_{j\notin \{i,i'\}}\la \sigma_j\ra_{\{i,i'\}}^2=D_{\{i,i'\}}.$$ 
    It follows that from the Lipschitz property of $\phi$ and the fact that $|\TH'(x)|\leq 1$, there exists a positive constant $L>0$ such that
	\begin{align*}
		\e\left| \e_{z}\phi \bigl(\TH(\beta z\sqrt{D_{\{i,i'\}}})\bigr) - \e_{z}\phi \bigl(\TH(\beta z\sqrt{q})\bigr)\right|
		\le &L \e|z|\cdot\e \bigl| \sqrt{D_{\{i,i'\}}} -\sqrt{q}  \bigr| \\
				\le &L \e|z|\cdot\bigl(\e \bigl| \sqrt{D_{\{i,i'\}}} -\sqrt{q}  \bigr|^2\bigr)^{1/2} \\
		\le &L\e |z|\cdot\bigl( \e \big| D_{\{i,i'\}}-q\big|\bigr)^{1/2}   \to 0,
	\end{align*} 
where the last inequality used the inequality $(\sqrt{x}-\sqrt{y})^2\leq |x-y|$ for any $x,y\geq 0.$ The same limit is also valid for $\psi.$
	Consequently, 
    \begin{align}\label{add:eq-3}
    \lim_{n\to\infty}\sup_{i,i'\in [n]:i\neq i}\e\bigl|\e_{i,i'}\phi(\la \sigma_i\ra)\psi(\la \sigma_{i'}\ra)-\e \phi\bigl(\TH\bigl(\beta z\sqrt{q}\bigr)\bigr)\cdot \e \psi\bigl(\TH\bigl(\beta z\sqrt{q}\bigr)\bigr)\bigr|=0.
    \end{align}
    Finally, write
    \begin{align}
    	\begin{split}\nonumber
    &\e\Bigl|\frac{1}{n}\sum_{i\in [n]}\la \sigma_i\ra-\e \TH(\beta z\sqrt{q})\Bigr|^2
    \end{split}\\
    \begin{split}\label{add:eq-4}
    &=\frac{1}{n^2}\sum_{i,i'\in [n]:i\neq i'}\e \Bigl[\e_{i,i'}\bigl[\la \sigma_i\ra\la \sigma_{i'}\ra\bigr]+\bigl(\e \TH(\beta z\sqrt{q})\bigr)\bigr)^2\\
    &\qquad-\e_{i,i'}\bigl[\la \sigma_i\ra\bigr]\e \TH(\beta z\sqrt{q})-\e_{i,i'}\bigl[\la \sigma_{i'}\ra\bigr]\e \TH(\beta z\sqrt{q})\Bigr]+O(n^{-1}),
    \end{split}
    \end{align}
where $O(n^{-1})$ comes from the total error of the main diagonal terms. 
From \eqref{add:eq-3}, the first term on the right can be handled by considering $\phi(x)=\psi(x)=x$, whereas the last two terms can be handled by setting $\phi(x)=x$ and $\psi(x)\equiv 1$. From these, the summation on the right-hand side of \eqref{add:eq-4} asymptotically vanishes. This completes our proof.
\end{proof}

\begin{proof}[\bf Proof of \eqref{assertion:1} in Proposition \ref{lem-1}:]
We argue by induction on $k\geq 2.$
	Consider $k=2$ and an arbitrary $\ell\geq 0.$ From Lemma~\ref{add:lem3}, 
	\begin{align}\label{add:eq-6}
R_S^{2}=\frac{1}{n}\sum_{i\notin S}\la \sigma_i\ra_S\nu _{S,i}^{[2]}&\asymp_1\frac{1}{n}\sum_{i\notin S}\Gamma\bigl(\rho_{S\cup\{i\}}^1;D_{S\cup\{i\}},E_{S\cup\{i\}}^{1}\bigr).
	\end{align}
 Now, from \eqref{Talagrand:eq2} and Lemma \ref{add:lem2},
	\begin{align}\label{add:eq-5}
	R_{S\cup\{i\}}^1=\frac{\sqrt{q}}{n}\sum_{j\notin S\cup\{i\}}\la\sigma_j\ra_{S\cup\{i\}}\asymp_1\frac{\sqrt{q}}{n}\sum_{j=1}^n\la\sigma_j\ra\asymp_1 \sqrt{q}\e\TH(\beta z\sqrt{q})=Q(\beta,h).
	\end{align}
    Since $
    |\rho_{S\cup \{i\}}^{1}| \le 1$ by  the Cauchy-Schwarz inequality, it follows that
    \begin{align}
    	\begin{split}\Bigl| \rho_{S\cup \{i\}}^{1} - q^{-1}Q(\beta, h) \Bigr| &=q^{-1}\Bigl|\rho_{S\cup \{i\}}^{1} \bigl(q-\sqrt{D_{S\cup \{i\}} E_{S\cup\{i\}}^{1}}\bigr) + R_{S\cup\{i\}}^{1} - Q(\beta, h) \Bigr| \\
    	&\le  q^{-1}\Bigl(\bigl|q-\sqrt{D_{S\cup \{i\}} E_{S\cup\{i\}}^{1}}\bigr| + \bigl| R_{S\cup\{i\}}^{1} - Q(\beta, h)\bigr|\Bigr). 
    	\end{split}\label{add:eq-7}
    \end{align} 
	Using this, \eqref{assertion:2}, and \eqref{add:eq-5}, we have that uniformly in $(i,S)$ with $|S|=\ell$ and $i\notin S,$ 
	\begin{align}\label{add:eq-9}
	\rho_{S\cup\{i\}}^1&\asymp_1 q^{-1}Q(\beta,h).
	\end{align}
    Consequently, plugging this and \eqref{assertion:2} into \eqref{add:eq-6} yields our assertion for $k=2.$
	Now assume that \eqref{assertion:1} is valid for some $k\geq 2.$ To show that it is also valid for $k+1$, again we use Lemma~\ref{add:lem3} to write that uniformly over all $(i,S)$ with $|S|=\ell$ and $i\notin S,$
	\begin{align}\label{add:eq23}
	\frac{1}{n}\sum_{i\notin S} \la \sigma_i\ra_S\nu_{S,i}^{[k+1]}\asymp_1\frac{1}{n}\sum_{i\notin S}\Gamma\bigl(\rho_{S\cup\{i\}}^k;D_{S\cup\{i\}},E_{S\cup\{i\}}^{k}\bigr).
	\end{align}
	Using the induction hypothesis and again \eqref{assertion:2} yields that uniformly over all $(i,S)$ with $|S|=\ell$ and $i\notin S,$
	\begin{align*}
	\e\bigl|D_{S\cup\{i\}}-q\bigr|^2,\e\bigl|E_{S\cup\{i\}}^k-q\bigr|^2&\to 0,\\
	\e\bigl|\rho_{S\cup\{i\}}^k-q^{-1}\Delta^{\circ(k-1)}\bigl(Q(\beta,h)\bigr)\bigr|^2&\to 0,
	\end{align*}
where the second display is argued in the same way as \eqref{add:eq-9} by using an analogous inequality of \eqref{add:eq-7},
 \begin{align*}
\Bigl| \rho_{S\cup \{i\}}^{k} - q^{-1}\Delta^{\circ(k-1)}\bigl(Q(\beta, h)\bigr) \Bigr| \leq  q^{-1}\Bigl(\bigl|q-\sqrt{D_{S\cup \{i\}} E_{S\cup\{i\}}^{k}}\bigr| + \bigl| R_{S\cup\{i\}}^{k} - \Delta^{\circ(k-1)}\bigl(Q(\beta, h)\bigr)\bigr|\Bigr). 
\end{align*} 
	Plugging the above limits into \eqref{add:eq23}, we see that \eqref{assertion:1} follows for $k+1$ and this completes our proof.
	\end{proof}

\subsection{Establishing Theorem \ref{thm0}}

First of all, from \cite[Proposition 1.6.8]{Tal111} and our assumption \eqref{shc}, we readily see that the free energy corresponding to the Hamiltonian of the SK model converges to the replica-symmetric solution \eqref{freeenergy}.
On the other hand, Toninelli \cite{Ton02} showed that this limit is valid only if $(\beta,h)$ lies inside the AT line in the sense that \eqref{ATline} is valid. Hence, for the rest of the proof, we shall assume that \eqref{ATline} is in force.

	Next, write
	\begin{align*}
	\e\bigl\|\la \sigma\ra-m^{[k]}\bigr\|^2&\leq 2\e\bigl\|\la\sigma\ra-\nu ^{[k]}\bigr\|^2+2\e \|\nu ^{[k]}-m^{[k]}\|^2.
	\end{align*}
	Here, the second term vanishes as $n\to\infty$ by \eqref{thm1:eq1}; the first term can be written as
	\begin{align*}
	\e\bigl\|\la\sigma\ra-\nu ^{[k]}\bigr\|^2&=\e\|\la \sigma\ra\|^2+\e\bigl\|\nu ^{[k]}\bigr\|^2-2\e\bigl \la \la \sigma\ra,\nu ^{[k]}\bigr\ra\\
	&=\e D_\emptyset+\e E_\emptyset^k-2\e R_\emptyset^{k}.
	\end{align*}
	From Proposition \ref{lem-1}, for any $k\geq 2,$
	\begin{align}\label{add:eq-8}
	\lim_{n\to\infty}	\e\bigl\|\la\sigma\ra-\nu ^{[k]}\bigr\|^2&=2q-2\Delta^{\circ(k-1)}\bigl(Q(\beta,q)\bigr).
	\end{align}

It remains to show that the right-hand side of \eqref{add:eq-8} converges to zero as $k\to\infty$ or equivalently,
\begin{align}\label{limit}
	\lim_{k\to\infty}\Delta^{\circ k}(Q(\beta,q))=q.
\end{align} 
From \eqref{delta}, 
	\begin{align*}
	\Delta(t)=\e \TH\bigl(\beta z\sqrt{|t|}+\beta z_1\sqrt{q-|t|}\bigr)\TH\bigl(\beta \mbox{\rm sign}(t) z\sqrt{|t|}+\beta z_2\sqrt{q-|t|}\bigr),\,\,t\in [-q,q].
	\end{align*} 
This function maps $[-q,q]$ into $[-q,q]$ since from the Cauchy-Schwarz inequality,
\begin{align*}
|\Delta(t)|\leq \e\TH^2(\beta z \sqrt{q})=q,\,\,\forall t\in [-q,q].
\end{align*}
In addition, $\Delta$ has a fixed point at $q$ since $\Delta(q)=\e \TH^{2}(\beta z \sqrt{q} )=q$. By
using Gaussian integration by parts and noting that $\tanh'=1/\cosh^2$, for any $t\in [-q, q]$,
	\begin{align*}
	\Delta'(t)&=\beta^2\e \frac{1}{\cosh^2\bigl(\beta z\sqrt{|t|}+\beta z_1\sqrt{q-|t|}+h\bigr)}\frac{1}{\cosh^2\bigl(\beta\mbox{\rm sign}(t)z\sqrt{|t|}+\beta z_2\sqrt{q-|t|}+h\bigr)}.
	\end{align*}
Consequently, from the Cauchy-Schwarz inequality  and the validity of \eqref{ATline}, for any $t\in (-q,q),$
	\begin{align}\label{diffD}
	\Delta'(t)&<\beta^2\e \frac{1}{\cosh^4\bigl(\beta z\sqrt{q}+h\bigr)}\leq 1.
	\end{align} 
	Now note that since $\Delta(q)=q$, if $\Delta(t)=t$ for some $t\in [-q,q),$ then from the mean value theorem, there exists some $t'\in (t,q)$ such that $\Delta'(t')=1$, which contradicts \eqref{diffD}. From this and noting that $\Delta(0)>0$, we must have that $t<\Delta(t)$ for all $t\in [-q,q).$ Consequently, 
	 \begin{align*}
	 Q(\beta,h)<\Delta(Q(\beta,h))
	 \end{align*}
 and since obviously $\Delta'(t)>0$ for all $t\in [-q,q]$, 
     \begin{align*}
     	\Delta^{\circ k}(Q(\beta,h))<\Delta^{\circ (k+1)}(Q(\beta,h)),\,\,\forall k\geq 1.
     \end{align*}
 Hence, $\lim_{k\to\infty}\Delta^{\circ k}(Q(\beta,h))$ exists and this limit must be a fixed point of $\Delta$ and then be equal to $q$, establishing \eqref{limit}. Our proof is complete.

\section{Proof of Theorem \ref{thm2}}\label{sec4.1}

Recall the vector $(W_k,W_{k-1},\ldots,W_1)$ from \eqref{lem2:eq3}. Consider an arbitrary bounded Lipschitz function $\psi:\mathbb{R}^{k+1}\to \mathbb{R}.$ We argue by induction on $k\geq 0$ that
\begin{align}\label{add:proof:eq1}
\lim_{n\to \infty}\e\Bigl|\frac{1}{n}\sum_{i=1}^n\psi\bigl(w_i^{[k]},w_{i}^{[k-1]},\ldots,w_i^{[0]}\bigr)-\e \psi\bigl(W_k,W_{k-1},\ldots,W_0\bigr)\Bigr|=0.
\end{align}
Obviously, the assertion is valid if $k=0,$ since the empirical measure of $w^{[0]}$ converges weakly to $W_0$.
Assume that the above statement is valid up to certain $k\geq 0.$ 
Recall from Proposition \ref{lem3} that for all $0\leq \ell\leq k,$ 
\begin{align*}
w_1^{[\ell+1]}&\asymp_1 w_{\{2\},1}^{[\ell+1]}=\frac{1}{\sqrt{n}}\sum_{j\neq 1,2}a_{1j}f_{\ell}\bigl(w_{\{1,2\},j}^{[\ell]}\bigr),\\
w_{2}^{[\ell+1]}&\asymp_1 w_{\{1\},2}^{[\ell+1]}=\frac{1}{\sqrt{n}}\sum_{j\neq 1,2}a_{2j}f_{\ell}\bigl(w_{\{1,2\},j}^{[\ell]}\bigr).
\end{align*}
Since the first and second rows and columns of $A_n$ are excluded in all $w_{\{1,2\},j}^{[\ell]}$ for all $j\ne 1,2$ and $0\leq \ell\leq k$, it follows  that
\begin{align*}
\bigl(w_{\{2\},1}^{[k+1]},w_{\{2\},1}^{[k]},\ldots,w_{\{2\},1}^{[0]}\bigr)\,\,\mbox{and}\,\,
\bigl(w_{\{1\},2}^{[k+1]},w_{\{1\},2}^{[k]},\ldots,w_{\{1\},2}^{[0]}\bigr)
\end{align*}
are independent conditioning on $(a_{i,j})_{i,j\neq 1,2}$ and each of them is jointly centered Gaussian with covariance, by the induction hypothesis, for $0\leq a,b\leq k,$
\begin{align*}
\e_{1}w_{\{2\},1}^{[a+1]}w_{\{2\},1}^{[b+1]}&=\frac{1}{n}\sum_{j\neq 2}f_{a}\bigl(w_{\{1,2\},j}^{[a]}\bigr)f_{b}\bigl(w_{\{1,2\},j}^{[b]}\bigr)\asymp_1\frac{1}{n}\sum_{j=1}^nf_{a}\bigl(w_{j}^{[a]}\bigr)f_{b}\bigl(w_{j}^{[b]}\bigr)\asymp_1\e f_a(W_a)f_b(W_b),\\
\e_{2} w_{\{1\},2}^{[a+1]}w_{\{1\},2}^{[b+1]}&=\frac{1}{n}\sum_{j\neq 1}f_{a}\bigl(w_{\{1,2\},j}^{[a]}\bigr)f_{b}\bigl(w_{\{1,2\},j}^{[b]}\bigr)\asymp_1\frac{1}{n}\sum_{j=1}^nf_{a}\bigl(w_{j}^{[a]}\bigr)f_{b}\bigl(w_{j}^{[b]}\bigr)\asymp_1 \e f_a(W_a)f_b(W_b).
\end{align*}
From these, for any two bounded Lipschitz functions $\phi_{1},\phi_{2}:\mathbb{R}^{k+2}\to \mathbb{R}$,
\begin{align*}
&\lim_{n\to\infty}\e \bigl[\phi_{1}\bigl(w_1^{[k+1]},w_{1}^{[k]},\ldots,w_{1}^{[0]}\bigr)\phi_{2}\bigl(w_2^{[k+1]},w_{2}^{[k]},\ldots,w_{2}^{[0]}\bigr)\bigr]\\
&=\lim_{n\to\infty}\e \bigl[\phi_{1}\bigl(w_{\{2\},1}^{[k+1]},w_{\{2\},1}^{[k]},\ldots,w_{\{2\},1}^{[0]}\bigr)\phi_{2}\bigl(w_{\{1\},2}^{[k+1]},w_{\{1\},2}^{[k]},\ldots,w_{\{1\},2}^{[0]}\bigr)\bigr]\\
&=\lim_{n\to\infty}\e \Bigl[\e_{1}\bigl[\phi_{1}\bigl(w_{\{2\},1}^{[k+1]},w_{\{2\},1}^{[k]},\ldots,w_{\{2\},1}^{[0]}\bigr)\bigr]
\e_2\bigl[\phi_{2}\bigl(w_{\{1\},2}^{[k+1]},w_{\{1\},2}^{[k]},\ldots,w_{\{1\},2}^{[0]}\bigr)\bigr]\Bigr]\\
&=\e\bigl[\phi_{1}(W_{k+1},W_k,\ldots,W_0)\bigr]\e\bigl[\phi_{2}(W_{k+1},W_k,\ldots,W_0)\bigr].
\end{align*}
Finally, by the symmetry among sites and the above limit, we arrive at
\begin{align}
\notag
&\lim_{n\to\infty}\e\Bigl[\Bigl(\frac{1}{n}\sum_{i=1}^n\phi_{1}(w_i^{[k+1]},w_{i}^{[k]},\ldots,w_{i}^{[0]})\Bigr)
\Bigl(\frac{1}{n}\sum_{i=1}^n\phi_{2}(w_i^{[k+1]},w_{i}^{[k]},\ldots,w_{i}^{[0]})\Bigr)\Bigr]\\
\notag
&=\lim_{n\to\infty}\e\bigl[ \phi_{1}\bigl(w_1^{[k+1]},w_{1}^{[k]},\ldots,w_{1}^{[0]}\bigr)\phi_{2}\bigl(w_2^{[k+1]},w_{2}^{[k]},\ldots,w_{2}^{[0]}\bigr)\bigr]\\
\label{eqn:induct-next-limit}
&=\e\bigl[\phi_{1}(W_{k+1},W_k,\ldots,W_0)\bigr]\e\bigl[\phi_{2}(W_{k+1},W_k,\ldots,W_0)\bigr].
\end{align}
To validate \eqref{add:proof:eq1} for the  $k+1$ case, or equivalently, 
	\[
	\lim_{n\to\infty}\e\Bigl[\frac{1}{n}\sum_{i=1}^n\psi\bigl(w_i^{[k+1]},w_{i}^{[k]},\ldots,w_i^{[0]}\bigr)-\e \psi\bigl(W_{k+1},W_{k},\ldots,W_0\bigr)\Bigr]^{2}=0,
	\]
	we expand the square here and apply \eqref{eqn:induct-next-limit} twice for the choices $\phi_{1}=\phi_{2}=\psi$ and $\phi_{1}=\psi, \phi_{2}\equiv1$. The resulting limits ultimately cancel each other.

\section{Moment controls}\label{sec4}

This section is a preparation for the proof of Theorem \ref{thm1}. 

\subsection{Main estimates}

Let $m\geq 0$. For $0\leq k\leq n-1,$ let $\mathcal{B}_{k,n}(m)$ be the set of all $(P,S,i)$ for $P$ being a multiset of elements in $\{(i,j):1\leq i<j\leq n\}$ with $|P|=m$ counting multiplicities and $i\in[n]$  and $S\in [n]_k$ satisfying that $i\notin S$.  Recall the definition of $w_{S,i}^{[k]}$ from \eqref{iteration}. Throughout this section, we write $$w_{S,i}^{[k]}=w_{S,i}^{[k]}(A)$$ to emphasize its dependence on the Gaussian matrix $A_n.$ Also, recall that $A_n$ is symmetric. For any $P=\{(i_1,j_1),\ldots,(i_m,j_m)\}$ and smooth $F$ defined on the space of $n\times n$ symmetric matrices, we adapt the notation
$$
\partial_PF(A)=\partial_{a_{i_rj_r},a_{i_{r-1}j_{r-1}},\ldots,a_{i_1j_1}}F(A),
$$
the partial derivatives of $F$ in the variables ${a_{i_rj_r},a_{i_{r-1},j_{r-1}},\ldots,a_{i_1j_1}}.$  
The following propositions control the moments of the partial derivatives of $w_{S,i}^{[k]}(A)$ in the entries of $A_n$.

\begin{proposition}\label{lem1}
	For any $k\geq 0$, $m\geq 0$, and $p\geq 1$, there exists a constant $W_{k,m,p}>0$ such that for all $n\geq k+1,$
	\begin{align}
	\begin{split}\label{lem1:eq1}
	\sup_{(P,S,i)\in\mathcal{B}_{k,n}(m)}\bigl(\e \bigl|\partial_Pw_{S,i}^{[k]}(A)\bigr|^p\bigr)^{1/p}&\leq \frac{W_{k,m,p}}{n^{m/2}}
	\end{split}
	\end{align}
	and for any smooth function $\zeta$ with bounded derivatives of all orders, there exists a constant $W_{k,m,p,\zeta}>0$  such that for all $n\geq k+1$, 
	\begin{align}
	\begin{split}\label{lem1:eq2}
	\sup_{(P,S,i)\in\mathcal{B}_{k,n}(m)}\bigl(\e \bigl|\partial_P\bigl(\zeta\bigl(w_{S,i}^{[k]}(A)\bigr)\bigr)\bigr|^p\bigr)^{1/p}&\leq \frac{W_{k,m,p,\zeta}}{n^{m/2}}.
	\end{split}
	\end{align}
\end{proposition}

\begin{proposition}\label{prop3}
	Let $\zeta:\mathbb{R}\to\mathbb{R}$ be a smooth function with bounded derivatives of all orders. For any $k\geq 0,$ $m\geq 0$, and $p\geq 1,$ there exist a constant $ W_{k,m,p,\zeta}' >0 $ such that for any $n\geq k+1,$
	\begin{align*}
	\sup_{}\Bigl(\e\Bigl|\partial_P\Bigl(\zeta\Bigl(\frac{1}{\sqrt{n}}\sum_{j\neq i,i'}a_{ij}f_k\bigl(w_{\{i\},j}^{[k]}(A)\bigr)\Bigr)\Bigr)\Bigr|^p\Bigr)^{1/p}\leq \frac{W_{k,m,p,\zeta}'}{n^{m/2}},
	\end{align*}
	where the supremum is taken over all $P$'s, collections of pairs from $\{(i,j):1\leq i<j\leq n\}$ with $|P|=m$ counting multiplicities and $i,i'\in [n]$ with $i\neq i'.$ 
\end{proposition}

These propositions  say that each partial derivative essentially brings up a factor $1/\sqrt{n}$. Indeed, in view of the definition of $w_{S,i}^{[k]}(A)$, although its partial derivatives involve a huge number of multiplications of the entries $a_{ij}/\sqrt{n}$, it turns out that due to the independence of the entries $a_{ij}$ for $i<j$, it can be shown that the total error  introduced by these multiplications is negligible resulting in the desired bounds. Notably similar inequalities were also established in \cite{ChenLam20} in the setting that the entries are independent and match the first and second moments of those of a standard Gaussian random variable.
\subsection{Proof of Proposition \ref{lem1}}

Before turning to the proof of Proposition \ref{lem1}, we prepare two lemmas. Let $r\in [n]$ and $a=(a_1,\ldots,a_r)$ be i.i.d. standard Gaussian random variables.  Let $$
F_1(x),\ldots, F_r(x):\mathbb{R}^r\to \mathbb{R}\,\,\mbox{for $x=(x_1,\ldots,x_r)$}$$
be random smooth functions, whose randomness are independent of $a$. For any $m\geq 0,$ denote by $P$, a multiset of elements from $\{1,\ldots,m\}$  and by $|P|$, the number of elements in $P$ counting multiplicities.  Denote by $\partial_PF_i$ the partial derivatives of $F_i$ with respect to the variables $x_j$ for $j\in P$ counting multiplicities. 

\begin{lemma}
	\label{lem2} Assume that for any $m\geq 0$ and even $p\geq 2$, there exists a constant $K_{m,p}>0$ such that
	\begin{align*}
	\sup_{j\in [r],|P|=m}\bigl(\e |\partial_PF_{j}(a)|^{p}\bigr)^{1/p}&\leq \frac{K_{m,p}}{n^{m/2}},\,\,\forall n\geq r.
	\end{align*}
	Then for any $m\geq 0$ and any even integer $p\geq 2,$ there exists a constant $K_{m,p}'>0$ independent of $n$ such that
	\begin{align}\label{lem2:eq1}
	\sup_{|P|=m}\Bigl(\e \Bigl|\frac{1}{\sqrt{n}}\sum_{j=1}^ra_{ij}\partial_PF_j(a)\Bigr|^p\Bigr)^{1/p}&\leq \frac{K_{m,p}'}{n^{m/2}},\,\,\forall n\geq r.
	\end{align}
\end{lemma}

\begin{proof} Let $p\geq 2$ be even. Let $m\geq 0$ and $P$ with $|P|=m$ be fixed. Write
	\begin{align*}
	\e \Bigl|\frac{1}{\sqrt{n}}\sum_{j=1}^ra_{j}\partial_PF_j(a)\Bigr|^p&=\frac{1}{n^{p/2}}\sum_{j_1,\ldots,j_p\in [r]}\e \bigl[a_{j_1}\cdots a_{j_p}L_{j_1,\ldots,j_p}(a)\bigr],	
	\end{align*}
	where 
	$$
	L_{j_1,\ldots,j_p}(a)=\prod_{s=1}^p \partial_PF_{j_s}(a).
	$$
	For $0\leq d\leq p,$ let $\mathcal{I}_d$ be the collection of all $(j_1,\ldots,j_p)\in [r]^p$ so that there are exactly $d$ indices in this vector that appear once in the list. Note that there exists a constant $C_{d,p}>0$ such that
	\begin{align}\label{lem3:proof:eq1}
	|\mathcal{I}_d|&\leq C_{d,p} n^d\cdot n^{\lfloor (p-d)/2\rfloor} ,
	\end{align}
	where $\lfloor t\rfloor$ is the largest integer less than or equal to $t.$
	Now we control  $\e \bigl[a_{j_1}\cdots a_{j_r}L_{j_1,\ldots,j_p}(a)\bigr].$ For any $(j_1,\ldots,j_p)\in \mathcal{I}_d$, if $j_1',\ldots,j_d'$ are those indices that appear once in $(j_1,\ldots,j_p)$, then from the Gaussian integration by parts, we have that
	\begin{align*}
	\e \bigl[a_{j_1}\cdots a_{j_r}L_{j_1,\ldots,j_p}(a)\bigr]&=\e \Bigl(\prod_{j\neq j_1',\ldots,j_d'}a_j \Bigr)\partial_{x_{j_1'}}\cdots \partial_{x_{j_d}'}L_{j_1,\ldots,j_p}(a)\\
	&\leq \Bigl(\e \Bigl(\prod_{j\neq j_1',\ldots,j_d'}a_j \Bigr)^2\Bigr)^{1/2}\e\bigl[\bigl|\partial_{x_{j_1'}}\cdots \partial_{x_{j_d}'}L_{j_1,\ldots,j_p}(a)\bigr|^2\bigr]^{1/2}.
	\end{align*}
	Here the first term in the last line is bounded above by $\bigl(\e |z|^{2p}\bigr)^{1/2}$. As for the second term, using the product rule, we readily write
	\begin{align*}
	\partial_{x_{j_1'}}\cdots\partial_{x_{j_d}'}L_{j_1,\ldots,j_p}(a)&=\sum \partial_{P_1}\bigl(\partial_PF_{j_1}(a)\bigr)\cdots \partial_{P_p}\bigl(\partial_PF_{j_p}(a)\bigr),
	\end{align*}
	where the sum is over all disjoint $P_1,\ldots,P_p$ with $\cup_{s=1}^pP_s=\{j_1',\ldots,j_d'\}.$ From the given assumption,
	\begin{align*}
	\bigl(	\e\bigl|\partial_{x_{j_1'}}\cdots\partial_{x_{j_d}'}L_{j_1,\ldots,j_p}(a)\bigr|^2\bigr)^{1/2}&\leq \sum \bigl(\e\bigl|\partial_{P_1}\bigl(\partial_PF_{j_1}(a)\bigr)\cdots \partial_{P_p}\bigl(\partial_PF_{j_p}(a)\bigr)\bigr|^2\bigr)^{1/2}\\
	&\leq \sum \prod_{s=1}^p\bigl(\e\bigl|\partial_{P_s}\bigl(\partial_PF_{j_s}(a)\bigr)\bigr|^{2p}\bigr)^{1/2p}\\
	&\leq p^d\prod_{s=1}^p \frac{\max_{0\leq r\leq d}K_{r+m,2p}}{n^{(|P_s|+m)/2}}\\
	&=\frac{1}{n^{(d+pm)/2}}p^d\bigl(\max_{0\leq r\leq d}K_{r+m,2p}\bigr)^p.
	\end{align*}
	Using this and \eqref{lem3:proof:eq1}, our proof is completed since
	\begin{align*}
	&\e \Bigl|\frac{1}{\sqrt{n}}\sum_{j=1}^ra_{j}\partial_PF_j(a)\Bigr|^p\\
	&\leq \frac{1}{n^{p/2}}\cdot \sum_{d=0}^pC_{d,p}n^d\cdot n^{\lfloor (p-d)/2\rfloor}\cdot \frac{1}{n^{(d+pm)/2}}p^d\bigl(\max_{0\leq r\leq d}K_{r+m,2p}\bigr)^p\cdot\bigl(\e|z|^{2p}\bigr)^{1/2}\\
	&=\frac{\bigl(\e|z|^{2p}\bigr)^{1/2}}{n^{pm/2}} \sum_{d=0}^p\frac{1}{n^{(p-d)/{2}-\lfloor (p-d)/{2}\rfloor}}C_{d,p}p^d\bigl(\max_{0\leq r\leq d}K_{r+m,2p}\bigr)^p\\
	&\leq \frac{1}{n^{pm/2}} K_{m,p}',
	\end{align*}
	where 
	$$
	K_{m,p}':=\bigl(\e|z|^{2p}\bigr)^{1/2}\sum_{d=0}^pC_{d,p}p^d\bigl(\max_{0\leq r\leq d}K_{r+m,2p}\bigr)^p.
	$$
\end{proof}

The proof of Proposition \ref{lem1} is argued as follows. First of all, note that \eqref{lem1:eq2} follows from \eqref{lem1:eq1} by applying the chain rule and the H\"older inequality. To show \eqref{lem1:eq1}, we argue by induction over $k.$ Obviously \eqref{lem1:eq1} holds for $k=0$.
Assume that there exists some $k_0\geq 0$ such that the assertion is valid for all $0\leq k\leq k_0$, $m\geq 0,$ and $p\geq 1$. We need to show that  \eqref{lem1:eq1} is valid for $k=k_0+1$ and all $m\geq 0,$ and $p\geq 1$.  Let $m\geq 0$ and $p\geq 1.$ For $n\geq k_0+2$, fix $(P,S,i)\in \mathcal{B}_{k_0+1,n}(m)$. Recall that
$$
w_{S,i}^{[k_0+1]}(A)=\frac{1}{\sqrt{n}}\sum_{j\notin S\cup\{i\}}a_{ij}f_{k_{0}}\bigl(w_{S\cup\{i\},j}^{[k_0]}(A)\bigr). 
$$
Set 
$$
v_{S\cup\{i\},j}(A)=f_{k_{0}}\bigl(w_{S\cup\{i\},j}^{[k_0]}(A)\bigr).
$$
Write $P=\{(i_1,j_1),\ldots,(i_m,j_m)\}$. Note that $A_n$ is symmetric. A straightforward computation yields that
\begin{align}
\nonumber &\partial_Pw_{S,i}^{[k_0+1]}(A)\\
\begin{split}\label{lem1:proof:eq1}
&=\frac{1}{\sqrt{n}}\sum_{r=1}^m\sum_{j\notin S\cup\{i\}}\bigl(\delta_{i,i_r}\delta_{j,j_r}\partial_{P\setminus\{(i_r,j_r)\}}v_{S\cup\{i\},j_r}(A)+\delta_{j,i_r}\delta_{i,j_r}\partial_{P\setminus\{(i_r,j_r)\}}v_{S\cup\{i\},i_r}(A)\bigr)
\end{split}\\
\begin{split}\label{lem1:proof:eq2}
&+\frac{1}{\sqrt{n}}\sum_{j\notin S\cup\{i\}}a_{ij}\partial_{P}v_{S\cup\{i\},j}(A),
\end{split}
\end{align}
where $\delta_{i,i'}=1$ if $i=i'$ and $=0$ otherwise. 
Note that here for all ${j\notin S\cup\{i\}},$
\begin{align*}
&\bigl(\delta_{i,i_r}\delta_{j,j_r}\partial_{P\setminus\{(i_r,j_r)\}}v_{S\cup\{i\},j_r}(A)+\delta_{j,i_r}\delta_{i,j_r}\partial_{P\setminus\{(i_r,j_r)\}}v_{S\cup\{i\},i_r}(A)\bigr)\\
&=\left\{
\begin{array}{ll}
0,&\mbox{if $\delta_{i,i_r}\delta_{j,j_r}=0=\delta_{j,i_r}\delta_{i,j_r}$},\\
\partial_{P\setminus\{(i_r,j_r)\}}v_{S\cup\{i\},j_r}(A),&\mbox{if $\delta_{i,i_r}\delta_{j,j_r}=1$ and $\delta_{j,i_r}\delta_{i,j_r}=0$},\\
\partial_{P\setminus\{(i_r,j_r)\}}v_{S\cup\{i\},i_r}(A),&\mbox{if $\delta_{i,i_r}\delta_{j,j_r}=0$ and $\delta_{j,i_r}\delta_{i,j_r}=1.$}
\end{array}\right.
\end{align*}
To bound each term in \eqref{lem1:proof:eq1} and \eqref{lem1:proof:eq2}, note that from the validity of \eqref{lem1:eq1}  with $k=k_0$, by using chain rule and the H\"older inequality, for any $m\geq 0$ and $p\geq 1,$ there exists a constant $K_{m,p}>0$ independent of $S$ and $i$ such that
\begin{align}\label{add:eq--5}
\sup_{j\notin S\cup\{i\},|P|=m}\Bigl(\e\bigl|\partial_Pv_{S\cup\{i\},j}(A)\bigr|^{p}\Bigr)^{1/p}\leq \frac{K_{m,p}}{n^{m/2}},\,\,\forall n\geq k_0+2.
\end{align}
Consequently, \eqref{lem1:proof:eq1} is bounded above by
\begin{align}\label{lem1:proof:eq3}
\frac{m}{n^{1/2}}\cdot\frac{2K_{m-1,p}}{n^{(m-1)/2}}=\frac{2mK_{m-1,p}}{n^{m/2}},\,\,\forall n\geq k_0+2
\end{align} To handle \eqref{lem1:proof:eq2}, set
\begin{align*}
F_j(A)&=v_{S\cup\{i\},j}(A),\,\,j\notin S\cup\{i\}.
\end{align*}
Note that these functions satisfy the assumption in Lemma \ref{lem2} due to \eqref{add:eq--5}. By applying to \eqref{lem2:eq1} for the $2p$-norm, there exists a constant $K_{m,2p}'>0$ independent of $n$
such that
\begin{align*}
\Bigl(\e\Bigl|\frac{1}{\sqrt{n}}\sum_{j\notin S\cup\{i\}}a_{ij}\partial_PF_j(A)\Bigr|^{2p}\Bigr)^{1/2p}
&\leq \frac{K_{m,2p}'}{n^{m/2}},\,\,\forall n\geq k_0+2.
\end{align*}
Note that this bound is uniformly valid  over all $(P,S,i)\in \mathcal{B}_{k_0+1,n}(m)$. From Jensen's inequality,
\begin{align*}
\sup_{\mathcal{B}_{k_0+1,n}(m)}{\Bigl(\e\Bigl|\frac{1}{\sqrt{n}}\sum_{j\notin S\cup\{i\}}a_{ij}\partial_PF_j(A)\Bigr|^p\Bigr)^{1/p}}&\leq \frac{K_{m,2p}'}{n^{m/2}},\,\,\forall n\geq k_0+2.
\end{align*}
Plugging this and \eqref{lem1:proof:eq3} into \eqref{lem1:proof:eq1} and \eqref{lem1:proof:eq2} and  applying the Minkowski inequality, we obtain that for  all $m\geq 0$ and $p\geq 1,$
\begin{align*}
\sup_{(P,S,i)\in \mathcal{B}_{k_0+1,n}(m)}\bigl(\e\bigl|\partial_Pw_{S,i}^{[k_0+1]}(A)\bigr|^p\bigr)^{1/p}\leq \frac{2mK_{m-1,p}+K_{m,2p}'}{n^{m/2}},\,\,\forall n\geq k_0+2,
\end{align*} 
which implies that \eqref{lem1:eq1} holds for $k=k_0+1$ and this completes the proof of \eqref{lem1:eq1}.

\subsection{Proof of Proposition \ref{prop3}}

Since $\zeta$ has bounded derivatives of all orders, by the virtue of  the chain rule, it suffices to show that for any $m\geq 0$ and $p\geq 1,$ there exists a constant $C>0$ such that
\begin{align}\label{add:eq13}
\sup\Bigl(\e\Bigl|\partial_P\Bigl(\frac{1}{\sqrt{n}}\sum_{j\neq i,i'}a_{ij}f_{k}\bigl(w_{\{i\},j}^{[k]}(A)\bigr)\Bigr)\Bigr|^p\Bigr)^{1/p}\leq \frac{C}{n^{m/2}},\,\,\forall n\geq k+1,
\end{align}
where the supremum is taken over all $P$, sets of elements in $\{(i,j):1\leq i<j\leq n\}$, with $|P|=m$ counting multiplicities and $i,i'\in [n]$ with $i\neq i'.$ To prove this, in a similar manner as \eqref{lem1:proof:eq1} and \eqref{lem1:proof:eq2}, we readily compute that for $P=\{(i_1,j_1),\ldots,(i_m,j_m)\}$,
\begin{align}
\begin{split}\notag
\partial_P\Bigl(\frac{1}{\sqrt{n}}\sum_{j\neq i,i'}a_{ij}f_k\bigl(w_{\{i\},j}^{[k]}(A)\bigr)\Bigr)
\end{split}\\
\begin{split}\label{lem1:proof:eq1'}
&=\frac{1}{\sqrt{n}}\sum_{r=1}^m\sum_{j\neq i,i'}\Bigl(\delta_{i,i_r}\delta_{j,j_r}\partial_{P\setminus\{(i_r,j_r)\}}\bigl(f_k\bigl(w_{\{i\},j_r}^{[k]}(A)\bigr)\bigr)+\delta_{j,i_r}\delta_{i,j_r}\partial_{P\setminus\{(i_r,j_r)\}}
\bigl(f_k\bigl(w_{\{i\},i_r}^{[k]}(A)\bigr)\bigr)\Bigr)
\end{split}\\
\begin{split}\label{lem1:proof:eq2'}
&+\frac{1}{\sqrt{n}}\sum_{j\neq i,i'}a_{ij}\partial_{P}\bigl(f_k\bigl(w_{\{i\},j}^{[k]}(A)\bigr)\bigr).
\end{split}
\end{align}
Here, using \eqref{lem1:eq2}, the $p$-th moment of \eqref{lem1:proof:eq1'} is bounded above by
\begin{align}\label{add:eq12}
\frac{1}{\sqrt{n}}\sum_{r=1}^m\sup_{(P,S,i)\in\mathcal{B}_{k_0,n}(m-1)}\Bigl(\e\Bigl|\partial_{P}\bigl(f_k\bigl(w_{S,i}^{[k]}(A)\bigr)\bigr)\Bigr|^p\Bigr)^{1/p}\leq \frac{C_0}{n^{m/2}},\,\,\forall n\geq k+1
\end{align}
for some constant $C_0>0.$
As for \eqref{lem1:proof:eq2'}, we write
\begin{align*}
\frac{1}{\sqrt{n}}\sum_{j\neq i,i'}a_{ij}\partial_{P}\bigl(f_k\bigl(w_{\{i\},j}^{[k]}(A)\bigr)\bigr)=\frac{1}{\sqrt{n}}\sum_{j\neq i}a_{ij}\partial_{P}\bigl(f_k\bigl(w_{\{i\},j}^{[k]}(A)\bigr)\bigr)-\frac{1}{\sqrt{n}}a_{ii'}\partial_{P}\bigl(f_k\bigl(w_{\{i\},j}^{[k]}(A)\bigr)\bigr)
\end{align*}
and use the Minkowski, Jensen, and Cauchy-Schwarz inequalities to get
\begin{align*}
&\Bigl(\e\Bigl|\frac{1}{\sqrt{n}}\sum_{j\neq i,i'}a_{ij}\partial_{P}\bigl(f_k\bigl(w_{\{i\},j}^{[k]}(A)\bigr)\bigr)\Bigr|^p\Bigr)^{1/p}\\
&\leq \Bigl(\e\Bigl|\frac{1}{\sqrt{n}}\sum_{j\neq i}a_{ij}\partial_{P}\bigl(f_k\bigl(w_{\{i\},j}^{[k]}(A)\bigr)\bigr)\Bigr|^{2p}\Bigr)^{1/2p}+\frac{1}{\sqrt{n}}\bigl(\e|a_{ii'}|^{2p}\bigr)^{1/2p}\bigl(\e\bigl|\partial_{P}\bigl(f_k\bigl(w_{\{i\},i'}^{[k]}(A)\bigr)\bigr)\bigr|^{2p}\bigr)^{1/2p}.
\end{align*}
Here, from \eqref{lem1:eq2}, the second term is bounded above by $C_1/n^{(m+1)/2}.$ Using \eqref{lem1:eq2} again and Lemma~\ref{lem2} for the $2p$-norm, the first term is bounded above by $C_2/n^{m/2}.$ Note that $C_1,C_2 > 0$ are universal constants independent of $n\geq  k_0+1$ and $P$ with $|P|=m$, and $i,i'\in[n]$ with $i\neq i'.$
Combining these together, the $p$-th moment of \eqref{lem1:proof:eq2'} is bounded by $(C_1+C_2)/n^{m/2}$.  This and \eqref{add:eq12} complete  the proof of \eqref{add:eq13}.

\section{Proof of Theorem \ref{thm1}}\label{sec5}

Our proof is based induction argument on $k.$ Before we start the proof, we set up some notations. 

\begin{notation}\label{notation2}\rm For any $x\in \mathbb{R}^n$ and $B$ an $n\times n$ matrix, denote the $2$-to-$2$ operator norm of $B$ by $
	\|B\|=\sup_{\|x\|=1}\|Bx\|.
	$
	For any $n\geq 1,$ let $u^n=(u_i^n)_{i\in[n]}$ and $v^n=(v_i^n)_{i\in [n]}$ be two sequences of random variables and $S_n\subset [n]$, we say that 
	$
	u_{i}^n\asymp_2 v^n_i
	$ for all $i\in S_n$
	if there exists a constant $C>0$ such that all sufficiently large $n,$
	\begin{align*}
	\sup_{i\in S_n}\e \bigl|u_i^n-v_i^n\bigr|^2\leq \frac{C}{n}.
	\end{align*}
	In addition, we say that $u^n\asymp_2 v^n$ if there exists a constant $C>0$ such that for all sufficiently large $n,$ $u_i^n\asymp_2 v_i^n$ for all $i\in [n].$
	For notational convenience, whenever there is no ambiguity, we will ignore the dependence on $n$ in these definitions.
\end{notation}

\subsection{An example}\label{sec:example}

To facilitate our proof, we argue that $w^{[2]}\asymp_2 u^{[2]}$ in this subsection. Note that $a_{ii}=0.$ Recall 
\begin{align}\label{add:eq--1}
u_i^{[2]}&=\frac{1}{\sqrt{n}}\sum_{j=1}^n a_{ij}f_1(u_j^{[1]})-\Bigl(\frac{1}{n}\sum_{j=1}^n f_1'(u_j^{[1]})\Bigr)f_0(u_i^{[0]}),\,\,i\in[n].
\end{align}
Fix $i\in [n]$. For each $j\in [n]$ with $j\neq i$, write
\begin{align*}
u_j^{[1]}&=\frac{1}{\sqrt{n}}\sum_{l\neq j}a_{jl}f_0(u_l^{[0]})=\frac{1}{\sqrt{n}}\sum_{l\neq i,j}a_{jl}f_0(u_l^{[0]})+\frac{a_{ij}}{\sqrt{n}}f_0(u_i^{[0]}).
\end{align*}
From this, we can use the Taylor expansion to get that
\begin{align}\label{add:eq--3}
f_1(u_j^{[1]})&=f_1\Bigl(\frac{1}{\sqrt{n}}\sum_{l\neq i,j}a_{jl}f_0(u_l^{[0]})\Bigr)+\frac{a_{ij}}{\sqrt{n}}f_1'\Bigl(\frac{1}{\sqrt{n}}\sum_{l\neq i,j}a_{jl}f_0(u_l^{[0]})\Bigr)f_0(u_i^{[0]})+\frac{O(a_{ij}^2)}{n}.
\end{align}
It follows that
\begin{align}
\begin{split}\notag
\frac{1}{\sqrt{n}}\sum_{j=1}^n a_{ij}f_1(u_j^{[1]})&\asymp_2 \frac{1}{\sqrt{n}}\sum_{j\neq i}a_{ij}f_1\Bigl(\frac{1}{\sqrt{n}}\sum_{l\neq i,j}a_{jl}f_0(u_l^{[0]})\Bigr)\\
&+\Bigl[\frac{1}{n}\sum_{j\neq i}a_{ij}^2f_1'\Bigl(\frac{1}{\sqrt{n}}\sum_{l\neq i,j}a_{jl}f_0(u_l^{[0]})\Bigr)\Bigr]f_0(u_i^{[0]})
\end{split}\\
\begin{split}\label{add:eq--2}
&= w_i^{[2]}+\Bigl[\frac{1}{n}\sum_{j\neq i}a_{ij}^2f_1'\Bigl(\frac{1}{\sqrt{n}}\sum_{l\neq i,j}a_{jl}f_0(u_l^{[0]})\Bigr)\Bigr]f_0(u_i^{[0]}).
\end{split}
\end{align}
Here, note that for each $i\in [n]$, $\{a_{ij}:j\neq i\}$ is independent of $\{a_{jl}:j\neq i\,\,\mbox{and}\,\, l\neq i,j\}$. This implies that $\{a_{ij}:j\neq i\}$ is independent of 
\begin{align*}
f_1'\Bigl(\frac{1}{\sqrt{n}}\sum_{l\neq i,j}a_{jl}f_0(u_l^{[0]})\Bigr),\,\,\forall j\neq i.
\end{align*}
As a result, using $\e (a_{ij}^2-1)=0$ and $\e (a_{ij}^2-1)^2=2$ yields that
\begin{align*}
\e\Bigl|\frac{1}{n}\sum_{j\neq i}(a_{ij}^2-1)f_1'\Bigl(\frac{1}{\sqrt{n}}\sum_{l\neq i,j}a_{jl}f_0(u_l^{[0]})\Bigr)\Bigr|^2&=\frac{2}{n^2}\sum_{j\neq i}\e\Bigl|f_1'\Bigl(\frac{1}{\sqrt{n}}\sum_{l\neq i,j}a_{jl}f_0(u_l^{[0]})\Bigr)\Bigr|^2\leq \frac{2\|f_1'\|_\infty}{n},
\end{align*}
which means that for all $i\in [n]$,
\begin{align*}
\frac{1}{n}\sum_{j\neq i}a_{ij}^2f_1'\Bigl(\frac{1}{\sqrt{n}}\sum_{l\neq i,j}a_{jl}f_0(u_l^{[0]})\Bigr)&\asymp_2 \frac{1}{n}\sum_{j\neq i}f_1'\Bigl(\frac{1}{\sqrt{n}}\sum_{l\neq i,j}a_{jl}f_0(u_l^{[0]})\Bigr)\\
&\asymp_2 \frac{1}{n}\sum_{j=1}^nf_1'\Bigl(\frac{1}{\sqrt{n}}\sum_{l=1}^na_{jl}f_0(u_l^{[0]})\Bigr).
\end{align*}
Combining \eqref{add:eq--1} and \eqref{add:eq--2} together yields that
$
u^{[2]}\asymp_2 w^{[2]}.
$

The proof of the general case $u^{[k+1]}\asymp_2 w^{[k+1]}$ consists of three major steps. In the first step, using the Taylor expansion as \eqref{add:eq--3} combining with the the induction hypothesis, it can be shown that the correction can be canceled leading to
\begin{align}\label{add:eq--4}
u_i^{[k+1]}&\asymp_2 \frac{1}{\sqrt{n}}\sum_{j\neq i}a_{ij}f_k\Bigl(\frac{1}{\sqrt{n}}\sum_{l\neq i,j}a_{jl}f_{k-1}\bigl(w_{\{j\},l}^{[k-1]}\bigr)\Bigr),\,\,\forall i\in [n].
\end{align}
To complete the proof, it remains to show that the right-hand side is asymptotically $w_i^{[k+1]}.$ The real difficult here is that one has to delete the $i$-th row and column of $A_n$ from $w_{\{j\},l}^{[k-1]}$. Although it is known that $w_{\{j\},l}^{[k-1]}\asymp_2w_{\{i,j\},l}^{[k-1]}$ from Proposition \ref{lem3}, we can not simply replace $w_{\{j\},l}^{[k-1]}$ by $w_{\{i,j\},l}^{[k-1]}$ since the double linear summations in \eqref{add:eq--4} can possibly amplify the accumulated error between them. Fortunately since our iteration adapts self-avoiding paths, the total error remains controllable by a subtle second moment estimate between the right-hand side of \eqref{add:eq--4} and $w^{[k+1]}$, which will be carried out in our second and third steps.

We now perform our main proof in three major steps. For convenience, $C,C_0,C_1,\ldots, C',C'',\ldots $ are universal (positive) constants that do not depend on any $n$ and $i\in [n]$ and they might mean different constants from line to line.

\subsection{Step I: Cancellation of the correction term}

Obviously the assertion holds when $k=0.$ Assume that  it is valid up to some $k\geq 0.$ From \eqref{amp} and the triangle inequality,
\begin{align*}
&\Bigl\|u^{[k+1]}-\frac{1}{\sqrt{n}}A_nf_k(w^{[k]})-\Bigl(\frac{1}{n}\sum_{j=1}^nf_k'(w_j^{[k]})\Bigr)f_{k-1}(w^{[k-1]})\Bigr\|\\
&\leq \frac{1}{\sqrt{n}}\|A_n\|\|f_k(u^{[k]})-f_k(w^{[k]})\|\\
&+M_{k-1}^{(0)}\|f_k'(u^{[k]})-f_k'(w^{[k]})\|\\
&+M_k^{(1)}\|f_{k-1}(u^{[k-1]})-f_{k-1}(w^{[k-1]})\|,
\end{align*}
where $
M_\ell^{(r)}=\|f_\ell^{(r)}\|_\infty.
$
Since $\|A_n\|/\sqrt{n}$ is square-integrable and $f_k',f_{k-1}$ are Lipschitz, the induction hypothesis implies that 
\begin{align*}
u^{[k+1]}\asymp_2\frac{1}{\sqrt{n}}A_nf_k(w^{[k]})-\Bigl(\frac{1}{n}\sum_{j=1}^nf_k'(w_j^{[k]})\Bigr)f_{k-1}(w^{[k-1]}).
\end{align*}
The following lemma is a crucial step, which gets rid of the correction term.

\begin{lemma}\label{lem0} For all $n\geq k+2,$ we have that 
	\begin{align}\label{lem0:eq1}
	u_i^{[k+1]}&\asymp_2 \frac{1}{\sqrt{n}}\sum_{j\neq i}a_{ij}f_k\Bigl(\frac{1}{\sqrt{n}}\sum_{l\neq i,j}a_{jl}f_{k-1}\bigl(w_{\{j\},l}^{[k-1]}\bigr)\Bigr),\,\,\forall i\in [n].
	\end{align}
\end{lemma}


\begin{proof}
	For each fixed $i\in [n]$, write by Taylor's expansion with respect to $a_{ij},$
	\begin{align*}
	&f_k(w_j^{[k]})\\
	&=f_k\Bigl(\frac{1}{\sqrt{n}}\sum_{l\neq j}a_{jl}f_{k-1}\bigl(w_{\{j\},l}^{[k-1]}\bigr)\Bigr)\\
	&=f_k\Bigl(\frac{1}{\sqrt{n}}\sum_{l\neq i,j}a_{jl}f_{k-1}\bigl(w_{\{j\},l}^{[k-1]}\bigr)+\frac{a_{ij}}{\sqrt{n}}f_{k-1}\bigl(w_{\{j\},i}^{[k-1]}\bigr)\Bigr)\\
	&=f_k\Bigl(\frac{1}{\sqrt{n}}\sum_{l\neq i,j}a_{jl}f_{k-1}\bigl(w_{\{j\},l}^{[k-1]}\bigr)\Bigr)+\frac{a_{ij}}{\sqrt{n}}f_k'\Bigl(\frac{1}{\sqrt{n}}\sum_{l\neq i,j}a_{jl}f_{k-1}\bigl(w_{\{j\},l}^{[k-1]}\bigr)\Bigr)f_{k-1}\bigl(w_{\{j\},i}^{[k-1]}\bigr)+\frac{O(a_{ij}^2)}{n}.
	\end{align*}
	As a result,
	\begin{align}
	\begin{split}\notag
	u_i^{[k+1]}&\asymp_2 \frac{1}{\sqrt{n}}\sum_{j\neq i}a_{ij}f_k\Bigl(\frac{1}{\sqrt{n}}\sum_{l\neq i,j}a_{jl}f_{k-1}\bigl(w_{\{j\},l}^{[k-1]}\bigr)\Bigr)
	\end{split}\\
	\begin{split}\label{add:eq3}
	&+\frac{1}{n}\sum_{j\neq i}a_{ij}^2B_{ij}D_{ij}-\frac{1}{n}\sum_{j}B_jD_i,\,\,\forall i\in[n],
	\end{split}
	\end{align}
	where
	\begin{align*}
	B_{ij}&=f_k'\Bigl(\frac{1}{\sqrt{n}}\sum_{l\neq i,j}a_{jl}f_{k-1}\bigl(w_{\{j\},l}^{[k-1]}\bigr)\Bigr),\quad D_{ij}=f_{k-1}\bigl(w_{\{j\},i}^{[k-1]}\bigr),\\
	B_j&=f_k'(w_j^{[k]}),\quad D_i=f_{k-1}(w_i^{[k-1]}).
	\end{align*}
	To handle the last two summations, we first claim that
	\begin{align*}
	\sup_{i\in[n]}\e\Bigl|\frac{1}{n}\sum_{j\neq i}(a_{ij}^2-1)B_{ij}D_{ij}\Bigr|^2=O(1/n).
	\end{align*}
   For fixed $i$, write the expectation term as
	\begin{align}\label{add:eq1}
	\frac{1}{n^2}\sum_{j,j'\neq i:j\neq j'}\e \bigl[y_{ij}B_{ij}D_{ij}y_{ij'}B_{ij'}D_{ij'}\bigr]+\frac{1}{n^2}\sum_{j\neq i}\e\bigl[y_{ij}^2B_{ij}^2D_{ij}^2\bigr],
	\end{align}
	where $y_{ij}:=a_{ij}^2-1.$ Here, the second term is of order $O(1/n).$ To control the first term, observe that conditionally on $a_{rr'}$ for $(r,r')\notin \{(i,j),(j,i),(i,j'),(j',i)\}$, $y_{ij'}B_{ij}D_{ij}$ depends only $a_{ij'}=a_{ji'}$ and $y_{ij}B_{ij'}D_{ij'}$ depends only on $a_{ij}=a_{ji}$. It follows that
	\begin{align*}
	\e \bigl[y_{ij}B_{ij}D_{ij}y_{ij'}B_{ij'}D_{ij'}\bigr]&=\e \bigl[\bigl(y_{ij}B_{ij'}D_{ij'}\bigr)\bigl(y_{ij'}B_{ij}D_{ij}\bigr)\bigr]\\
	&=\e\bigl[\e_{a_{ij'}} \bigl[y_{ij}B_{ij'}D_{ij'}\bigr] \e_{a_{ij}}\bigl[y_{ij'}B_{ij}D_{ij}\bigr]\bigr],
	\end{align*}
	where $\e_{a_{ij}}$ is the expectation for $a_{ij}$ and $\e_{a_{ij'}}$ is the expectation for $a_{ij'}$.
	Now using the mean value theorem and Proposition \ref{lem3},
	\begin{align}
	\begin{split}\label{add:eq5}
	B_{ij}&\asymp_2 B_j\asymp_2 f_k'\bigl(w_{\{i,j'\},j}^{[k]}\bigr)=:B_{\{i,j'\},j},\\
	D_{ij}&\asymp_2 f_{k-1}\bigl(w_{\{j,j'\},i}^{[k-1]}\bigr)=:D_{\{j,j'\},i}.
	\end{split}
	\end{align}
	Write
	\begin{align*}
	\e_{a_{ij'}}\bigl[y_{ij'}B_{ij}D_{ij}\bigr]&=\e_{a_{ij'}}\bigl[y_{ij'}\bigl(B_{ij}-B_{\{i,j'\},j}\bigr)\bigl(D_{ij}-D_{\{j,j'\},i}\bigr)\bigr]\\
	&+\e_{a_{ij'}} \bigl[y_{ij'}\bigl(B_{ij}-B_{\{i,j'\},j}\bigr)D_{\{j,j'\},i}\bigr]\\
	&+\e_{a_{ij'}}\bigl[y_{ij'}B_{\{i,j'\},j}\bigl(D_{ij}-D_{\{j,j'\},i}\bigr)\bigr]\\
	&+\e_{a_{ij'}} \bigl[y_{ij'}B_{\{i,j'\},j} D_{\{j,j'\},i}\bigr].
	\end{align*}
	Note that $B_{\{i,j'\},j}$ and $D_{\{j,j'\},i}$ are both independent of $a_{ij'}$ so that  $	\e_{a_{ij'}}\bigl[y_{ij'}B_{\{i,j'\},j} D_{\{j,j'\},i}\bigr]=0.$
	Consequently, from the Cauchy-Schwarz inequality and \eqref{add:eq5}, there exists a constant $C_0>0$ such that
	\begin{align*}
	\bigl(\e\bigl(\e_{a_{ij'}}(y_{ij'}B_{ij}D_{ij})\bigr)^2\bigr)^{1/2}&\leq \frac{C_0}{\sqrt{n}}.
	\end{align*}
	The same inequality is also valid for $\bigl(\e\bigl(\e_{a_{ij}}(y_{ij}B_{ij'}D_{ij'})\bigr)^2\bigr)^{1/2}.$ Using the Cauchy-Schwarz inequality to the first summation of \eqref{add:eq1} completes the proof of our claim.
	
	Next, by the virtue of the above claim, we have
	\begin{align}\label{add:eq2}
	\frac{1}{n}\sum_{j\neq i}a_{ij}^2B_{ij}D_{ij}&\asymp_2 \frac{1}{n}\sum_{j\neq i}B_{ij}D_{ij}.
	\end{align}
	Write
	\begin{align*}
	\frac{1}{n}\sum_{j\neq i}\bigl(B_{ij}D_{ij}-B_jD_i\bigr)
	&=\frac{1}{n}\sum_{j\neq i}(B_{ij}-B_j)D_{ij}+\frac{1}{n}\sum_{j\neq i}(D_{ij}-D_i)B_j.
	\end{align*}
	Here since
	\begin{align*}
	\bigl|B_{ij}-B_j\bigr|&\leq \frac{C_1|a_{ij}|}{\sqrt{n}},
	\end{align*}
	it follows that
	\begin{align*}
	\e\Bigl|\frac{1}{n}\sum_{j\neq i}(B_{ij}-B_j)D_{ij}\Bigr|^2\leq \frac{C_2}{n}.
	\end{align*}
	On the other hand, by Proposition \ref{lem3},
	\begin{align*}
	\e\Bigl|\frac{1}{n}\sum_{j\neq i}(D_{ij}-D_i)B_j\Bigr|^2&\leq \frac{C_3}{n}.
	\end{align*}
	Putting these together yields that
	\begin{align*}
	\frac{1}{n}\sum_{j\neq i}\bigl(B_{ij}D_{ij}-B_jD_i\bigr)\asymp_2 0.
	\end{align*}
	From this and \eqref{add:eq2},
	\begin{align*}
	\frac{1}{n}\sum_{j\neq i}a_{ij}^2B_{ij}D_{ij}\asymp_2\frac{1}{n}\sum_{j\neq i}B_{j}D_{i}\asymp_2 \frac{1}{n}\sum_{j}B_{j}D_{i}.
	\end{align*}
	Hence, the last two summations in \eqref{add:eq3} cancels each other so that \eqref{lem0:eq1} follows.
\end{proof}

From Lemma \ref{lem0}, our proof of Theorem \ref{thm1} is complete if we can show that for all $i\in [n],$
\begin{align*}
\frac{1}{\sqrt{n}}\sum_{j\neq i}a_{ij}f_k\Bigl(\frac{1}{\sqrt{n}}\sum_{l\neq i,j}a_{jl}f_{k-1}\bigl(w_{\{j\},l}^{[k-1]}\bigr)\Bigr)\asymp_2 \frac{1}{\sqrt{n}}\sum_{j\neq i}a_{ij}f_k\Bigl(\frac{1}{\sqrt{n}}\sum_{l\neq i,j}a_{jl}f_{k-1}\bigl(w_{\{i,j\},l}^{[k-1]}\bigr)\Bigr)=w_{i}^{[k+1]}.
\end{align*}
Fix $i\in [n]$. For any $j\neq i,$ set
\begin{align*}
L_j&=f_k\Bigl(\frac{1}{\sqrt{n}}\sum_{l\neq i,j}a_{jl}f_{k-1}\bigl(w_{\{j\},l}^{[k-1]}\bigr)\Bigr),\\
K_{j}&=f_k\Bigl(\frac{1}{\sqrt{n}}\sum_{l\neq i,j}a_{jl}f_{k-1}\bigl(w_{\{i,j\},l}^{[k-1]}\bigr)\Bigr).
\end{align*}
For any two distinct indices $\tau,\iota \in [n]\setminus \{i\}$, if we condition on all $a_{rr'}$'s for $(r,r')\not \in \{(i,\tau), (i,\iota), (\tau, i), (\iota, i)\}$, then $L_{\tau}$ will only depend on $a_{i\iota} = a_{\iota i}$ and $L_{\iota}$ only depends on $a_{i\tau} = a_{\tau i}$. In addition, $(a_{ij})_{j\neq i}$ is independent of $K_{\tau}$ and $K_{\iota}$. It follows that   
\begin{align*}
\e \bigl[a_{i\tau}a_{i\iota} L_{\tau}L_{\iota}\bigr] &= \e \bigl[\e_{a_{i\tau}}\bigl[a_{i\tau}L_\iota\bigr]\e_{a_{i\iota}}\bigl[a_{i\iota}L_\tau\bigr]\bigr],\\
\e \bigl[a_{i\tau}a_{i\iota} L_{\tau}K_{\iota}\bigr] &=\e\bigl[a_{i\tau}\bigr]\e\bigl[a_{i\iota} L_{\tau}K_{\iota}\bigr] =0,\\
\e \bigl[a_{i\tau}a_{i\iota} K_{\tau}K_{\iota}\bigr]  &=\e \bigl[a_{i\tau}a_{i\iota}\bigr]\e\bigl[K_{\tau}K_{\iota}\bigr] =0,
\end{align*}
where recall that $\e_{a_{i\tau}}$ and $\e_{a_{i\iota}}$ are the expectations with respect to $a_{i\tau}$ and $a_{i\iota}$, respectively. From these, 
\begin{align}
\begin{split}\notag
&\e\Bigl|\frac{1}{\sqrt{n}}\sum_{j\neq i}a_{ij}\bigl(L_j-K_j\bigr)\Bigr|^2\\
&=\frac{1}{n}\sum_{\tau,\iota\neq i:\tau\neq \iota}\e \left[a_{i\tau}\bigl(L_{\iota}-K_{\iota}\bigr)a_{i\iota}\bigl(L_\tau-K_\tau\bigr)\right]+\frac{1}{n}\sum_{j\neq i} \e a_{ij}^2\bigl(L_j-K_j\bigr)^2\\
\end{split}\\
\begin{split}\label{add:eq4}
&=\frac{1}{n}\sum_{\tau,\iota\neq i:\tau\neq \iota}\e\bigl[\e_{a_{i\tau}}\bigl[a_{i\tau}L_\iota\bigr]\e_{a_{i\iota}}\bigl[a_{i\iota}L_\tau\bigr]\bigr]+\frac{1}{n}\sum_{j\neq i} \e a_{ij}^2\bigl(L_j-K_j\bigr)^2.
\end{split}
\end{align}
Our next two steps control these two summations.

\subsection{Step II: Diagonal case}

From the mean value theorem, the second summation of \eqref{add:eq4} can be handled by
\begin{align}
\begin{split}\notag
\frac{1}{n}\sum_{j\neq i} \e a_{ij}^2\bigl(L_j-K_j\bigr)^2&=\frac{1}{n}\sum_{j\neq i}\e\bigl(L_j-K_j\bigr)^2\\
&\leq \frac{C}{n}\sum_{j\neq i}\e\Bigl|\frac{1}{\sqrt{n}}\sum_{l\neq i,j}a_{jl}\bigl(f_{k-1}\bigl(w_{\{j\},l}^{[k-1]}\bigr)-f_{k-1}\bigl(w_{\{i,j\},l}^{[k-1]}\bigr)\bigr)\Bigr|^2\\
&=\frac{C}{n^2}\sum_{j\neq i}\sum_{l\neq i,j}\e \bigl|f_{k-1}\bigl(w_{\{j\},l}^{[k-1]}\bigr)-f_{k-1}\bigl(w_{\{i,j\},l}^{[k-1]}\bigr)\bigr|^2\\
&\leq \frac{C'}{n^2}\sum_{j\neq i}\sum_{l\neq i,j}\e \bigl|w_{\{j\},l}^{[k-1]}-w_{\{i,j\},l}^{[k-1]}\bigr|^2
\end{split}\\
\begin{split}\label{add:eq16}
&\leq \frac{C''}{n},
\end{split}
\end{align}
where the second equality used the fact that $(a_{jl})_{l\neq i,j}$ is independent of $(w_{\{j\},l}^{[k-1]})_{l\neq i,j}$ and $(w_{\{i,j\},l}^{[k-1]})_{l\neq i,j}$ and the last inequality used Proposition \ref{lem3}. 

\subsection{Step III: Off-diagonal case}
It remains to show that the first summation of \eqref{add:eq4} is of order $1/n$, which requires more subtle controls of the moments. Fix $i\in[n]$. Let $\tau,\iota\in [n]\setminus\{i\}$ and $\tau\neq \iota.$ First of all, we compute $\e_{a_{i\iota}}[a_{i\iota}L_{\tau}]$ using Gaussian integration by part and the chain rule as follows.  Write $L_{\tau}=f_{k}(\Delta_{\tau})$ for
\begin{align*}
\Delta_\tau&:=\frac{1}{\sqrt{n}}\sum_{\tau_{k-1}\neq i,\tau}a_{\tau \tau_{k-1}}f_{k-1}\bigl(w_{\{\tau \},\tau_{k-1}}^{[k-1]}\bigr).
\end{align*}
Here we would like to call the dummy variable in the summation $\tau_{k-1}$ as its subscript matches  the iteration number. This choice of dummy variable appears to be very convenient later when we need to look back into the $(k-1)$-th, $(k-2)$-th, $\ldots$, iterations. 

Since $\tau\ne \iota$ and $\tau_{k-1}\ne i, \tau$, we see that $a_{\tau\tau_{k-1}} \ne a_{i\iota}$ or $a_{\iota i}$. Applying Gaussian integration by parts yields
\[
\e_{a_{i\iota}}( a_{i\iota}L_{\tau} ) =\frac{1}{\sqrt{n}}\e_{a_{i\iota}} f_{k}'(\Delta_{\tau}) \sum_{\tau_{k-1} \ne i, \tau}a_{\tau \tau_{k-1}}\partial_{a_{i\iota}}f_{k-1}(w_{\{\tau\}, \tau_{k-1}}^{[k-1]}).
\]
In order to compute the partial derivative with respect to $a_{i\iota}$, we proceed by tracking back the iterations until either $a_{i\iota}$ or $a_{\iota i}$ appears at the $r$-th iteration for some $1\le r\le k-1$ (once either appears, neither of them will appear again in $w_{\{\tau,\tau_{k-1},\ldots,\tau_{s}\},\tau_{s-1}}^{[s-1]}$ for all $1\leq s\leq r$ due to the path self-avoiding property). Recall that
\begin{center}
\includegraphics[width=0.77\textwidth]{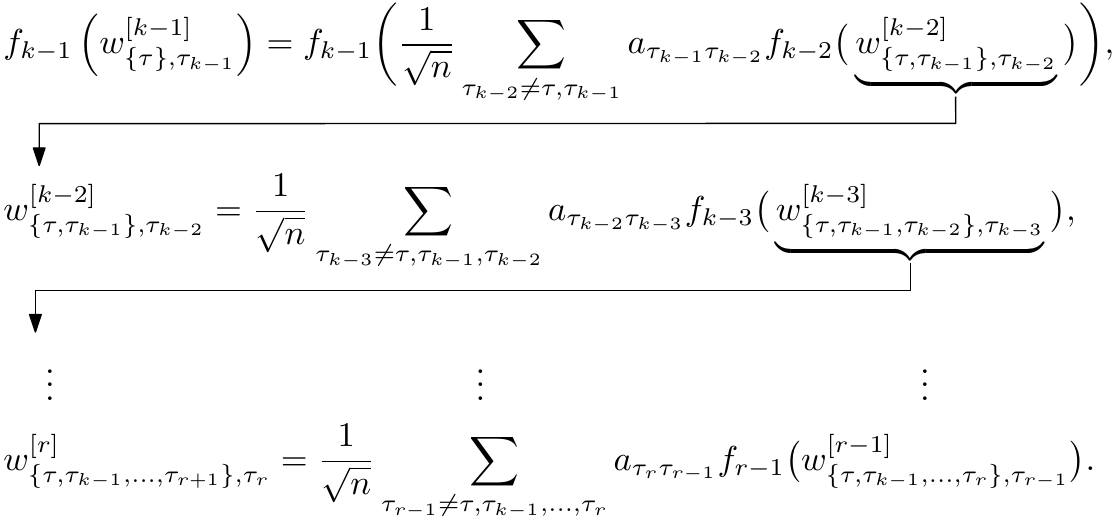}
\end{center}
As long as $(\tau_r,\tau_{r-1})$ equals $(i,\iota)$ or $(\iota,i)$ for the first time for some $1\leq r\leq k-1$, we have that for any $r\le s \le k-1$, 
\begin{align*}
&\partial_{a_{i\iota}}f_{s}\bigl(w_{\{\tau,\tau_{k-1},\ldots,\tau_{s+1}\},\tau_{s}}^{[s]}\bigr)\\
&=\left\{
\begin{array}{ll}
\displaystyle
\frac{1}{\sqrt{n}}f_{s}'\bigl(w_{\{\tau,\tau_{k-1},\ldots,\tau_{s+1}\},\tau_{s}}^{[s]}\bigr)\sum_{\tau_{s-1}}a_{\tau_{s}\tau_{s-1}}
\partial_{a_{i\iota}}f_{s-1}\bigl(w_{\{\tau,\tau_{k-1},\ldots,\tau_{s}\},\tau_{s-1}}^{[s-1]}\bigr), &\text{ if }s>r,\\
\displaystyle
\frac{1}{\sqrt{n}}f_{r}'\bigl(w_{\{\tau,\tau_{k-1},\ldots,\tau_{r+1}\},\tau_{r}}^{[r]}\bigr)f_{r-1}\bigl(w_{\{\tau, \tau_{k-1},\ldots,\tau_{r}\},\tau_{r-1}}^{[r-1]}\bigr),& \text{ if }s=r,
\end{array}
\right.
\end{align*}
where the summation is over all $\tau_{s-1}\ne \tau,  \tau_{k-1}, \ldots, \tau_{s}$. This computation suggests that the partial derivative at the $s$-th iteration for some $s>r$ must involve the partial derivative of the $(s-1)$-th iteration and a factor of $n^{-1/2}$ is brought up every time when the chain rule is applied, until $a_{i\iota}$ or $a_{\iota i}$ appears for the first time at the $r$-th iteration. This in total brings up a factor of $n^{-(k-(r-1))/2}$ and we finally get
\begin{align}
\begin{split}\notag
\e_{a_{i\iota}}\bigl[ a_{i\iota}L_{\tau} \bigr]=\sum_{r=1}^{k-1}\frac{1}{n^{\frac{k-(r-1)}{2}}}\e_{a_{i\iota}} \Bigl[\sum_{I_{\tau,r}\in \mathcal{I}_{\tau,r}}A_{I_{\tau,r}}F_{I_{\tau,r}}(A)\Bigr]\mathbbm 1_{\{(\tau_{r}, \tau_{r-1}) = (i, \iota) \text{ or } (\iota, i)\}},
\end{split}
\end{align} 
where $\mathcal{I}_{\tau, r}$ is the collection of all self-avoiding paths 
$$ I_{\tau,r}=(\tau_k,\tau_{k-1},\tau_{k-2},\ldots,\tau_{r},\tau_{r-1})\in [n]^{k-r+2}$$
of length $k-r+1$ starting from $\tau_k=\tau$ and satisfying $\tau_{k-1}\neq i$,  and  
\begin{align}
\begin{split}\label{add:eq10}
A_{I_{\tau, r}} &:=\prod_{s=r}^{k-1}a_{\tau_{s+1}\tau_{s}}, \\
F_{I_{\tau,r}}(A)&:=f_k'(\Delta_\tau)\Bigl(\prod_{s=r}^{k-1}f_{s}'\bigl(w_{\{\tau,\tau_{k-1},\ldots,\tau_{s+1}\},\tau_{s}}^{[s]}\bigr)\Bigr) f_{r-1}\bigl(w_{\{\tau,\tau_{k-1},\ldots,\tau_{r}\},\tau_{r-1}}^{[r-1]}\bigr).\\
\end{split}
\end{align}
Similarly,
\begin{align*}
\e_{a_{i\tau}} \bigl[a_{i\tau}L_\iota\bigr]&=\sum_{r=1}^{k-1}\frac{1}{n^{\frac{k-(r-1)}{2}}}\e_{a_{i\tau}} \Bigl[\sum_{I_{\iota,r}\in \mathcal{I}_{\iota,r}}A_{I_{\iota,r}}F_{I_{\iota,r}}(A)\Bigr]\mathbbm 1_{\{(\iota_{r}, \iota_{r-1}) = (i, \tau) \text{ or } (\tau, i)\}}.
\end{align*}
Now, from these
\begin{align}
\begin{split}\notag
&\e\bigl[\e_{a_{i\tau}}\bigl[a_{i\tau}L_\iota\bigr]\e_{a_{i\iota}}\bigl[a_{i\iota}L_\tau\bigr]\bigr]\\
\end{split}\\
\begin{split} \label{add:eq8}
=&\sum_{r,r'=1}^{k-1}\frac{1}{n^{k+1-\frac{r+r'}{2}}}\sum_{I_{\tau,r}\in \mathcal{I}_{\tau,r}}\sum_{I_{\iota,r'}\in \mathcal{I}_{\iota,r'}}\e\Bigl[A_{I_{\tau,r}}A_{I_{\iota,r'}}F_{I_{\tau,r}}(A)F_{I_{\iota,r'}}(A)\Bigr]\mathbbm 1_{\left\{\substack{(\tau_{r}, \tau_{r-1}) = (i, \iota) \text{ or } (\iota, i)\\
(\iota_{r'},\iota_{r'-1}) = (i,\tau) \text{ or } (\tau, i)}\right\}} ,
\end{split}
\end{align}
where the last equation used the fact that $A_{I_{\tau,r}}F_{I_{\tau,r}}(A)$ is independent of $a_{i\tau}$ and $A_{I_{\iota,r'}}F_{I_{\iota,r'}}(A)$ is independent of $a_{i\iota}.$
Each term in the summation of the last line is nonzero only if one of the following four cases is valid:
\begin{align*}
(A)&\quad(\tau_{r},\tau_{r-1})=(i,\iota),\,\,(\iota_{r'},\iota_{r'-1})=(i,\tau),\\
(B)&\quad(\tau_{r},\tau_{r-1})=(i,\iota),\,\,(\iota_{r'},\iota_{r'-1})=(\tau,i),\\
(C)&\quad(\tau_{r},\tau_{r-1})=(\iota,i),\,\,(\iota_{r'},\iota_{r'-1})=(i,\tau),\\
(D)&\quad(\tau_{r},\tau_{r-1})=(\iota,i),\,\,(\iota_{r'},\iota_{r'-1})=(\tau,i).
\end{align*}
Note that $\mathcal{I}_{\tau,r}$ and $\mathcal{I}_{\iota,r'}$ are collections of self-avoiding paths starting from $\tau$ and $\iota$, respectively. 
Let $\mathcal{I}_{\tau,\iota,r,r'}(s,t)$ be the collection of all pairs $(I_{\tau,r},I_{\iota,r'})\in \mathcal{I}_{\tau,r}\times \mathcal{I}_{\iota,r'}$ satisfying that
(i) one of $(A)-(D)$ holds,
(ii) there are exactly $s$ edges shared by $I_{\tau,r}$ and $I_{\iota,r'}$ disregard the direction, and
(iii) the number of (distinct) vertices appearing in the shared edges is equal to $t$.
See Figure \ref{fig1}(a) and (d) for two examples of pairs $(I_{\tau,r},I_{\iota,r'})$ in $\mathcal I_{\tau, \iota, k-5, k-6}(3,5)$ for $(i, \tau, \iota)=(1,2,4)$, where the shared edges are marked in blue.

Note that for $(I_{\tau,r},I_{\iota,r'})\in \mathcal{I}_{\tau,\iota,r,r'}(s,t)$, if the edge $(\tau_{r},\tau_{r-1})$ is shared in $I_{\iota,r'}$, it must imply that $\iota_{k-1}=i$ due to $(A)-(D)$, which contradicts  the definition of $\mathcal{I}_{\iota,r'}$ since $\iota_{k-1}\neq i.$ Hence, the last edges $(\tau_{r},\tau_{r-1})$ in $I_{\tau,r}$  and  $(\iota_{r'},\iota_{r'-1})$ in  $I_{\iota,r'}$ must not be among the shared edges. From this, to control the size of $\mathcal{I}_{\tau,\iota,r,r'}(s,t)$, it suffices to consider $s,t$ satisfying 
\begin{align}
\begin{split}\label{add:eq---1}
t = s = 0 \quad \text{or }\quad \left.\begin{array}{l}
\ \ \ \ \ 1\leq s\leq \min(k-r,k-r'),\\
s+1\leq t\leq \min\bigl(2s,k-r+1,k-r'+1\bigr).
\end{array}\right.
\end{split}
\end{align}
We then write
\begin{align}
\begin{split}\label{add:eq6}
&\sum_{I_{\tau,r}\in \mathcal{I}_{\tau,r}}\sum_{I_{\iota,r'}\in \mathcal{I}_{\iota,r'}}\e\Bigl[A_{I_{\tau,r}}A_{I_{\iota,r'}}F_{I_{\tau,r}}(A)F_{I_{\iota,r'}}(A)\Bigr]\mathbbm 1_{\left\{\substack{(\tau_{r}, \tau_{r-1}) = (i, \iota) \text{ or } (\iota, i)\\
(\iota_{r'},\iota_{r'-1}) = (i,\tau) \text{ or } (\tau, i)}\right\}}\\
&=\sum_{s,t}\sum_{(I_{\tau,r},I_{\iota,r'})\in \mathcal{I}_{\tau,\iota,r,r'}(s,t)}\e\Bigl[A_{I_{\tau,r}}A_{I_{\iota,r'}}F_{I_{\tau,r}}(A)F_{I_{\iota,r'}}(A)\Bigr],
\end{split}
\end{align}
where the first summation in the second line is over all $s,t$ satisfying \eqref{add:eq---1}.

Next, we further introduce the notation $\mathcal{I}_{\tau,\iota,r,r'}(s,t,\ell) \subset  \mathcal{I}_{\tau,\iota,r,r'}(s,t)$, where $\ell=0,1,2$ denotes the number of vertices in $\{\tau, \tau_{r}\}$ (or, equivalently, in $\{\iota,\iota_{r'}\}$; see Remark \ref{rmk:number-shared} below) that appear in the shared edges. In Figure \ref{fig1}, (a) and (d) are two examples in the same collection $\mathcal I_{\tau, \iota, k-5, k-6}(3,5)$ but with $\ell = 1$ and $\ell =2$, respectively.
Note that $\mathcal{I}_{\tau,\iota,r,r'}(s,t,\ell)=\emptyset$ if $\ell > t$. 

\begin{remark}
\label{rmk:number-shared}
\rm We claim that for any $(I_{\tau,r},I_{\iota,r'})\in \mathcal{I}_{\tau,\iota,r,r'}(s,t)$, the numbers of vertices in $\{\tau, \tau_{r}\}$ (denoted by $n_{1}$) and $\{\iota, \iota_{r'}\}$ (denoted by $n_{2}$) appearing in the shared edges must be the same, due to $(A)$-$(D)$. For symmetry, we only discuss the cases when $n_{1}<n_{2}$. 
\begin{itemize}
\item Case $n_{1}=0, n_{2}=1$. First of all, suppose $\iota$ is in a shared edge but $\tau, \tau_{r}$ and $\iota_{r'}$ are not. This immediately rules out $(C)$ and $(D)$ because in these two cases $\tau_{r}=\iota$. The cases $(A)$ and $(B)$ also can not occur. Indeed, if either $(A)$ or $(B)$ holds, then  this would force $(\tau_{r}, \tau_{r-1})$ (the last edge in $I_{\tau, r}$) to be a shared edge, a contradiction. Next, suppose that $\iota_{r'}$ is in a shared edge but $\tau, \tau_{r}$ and $\iota$ are not. We see that $(A)$, $(B)$ and $(D)$ can not occur because in these three cases, either $\iota_{r'}=\tau$ or $\iota_{r'}=\tau_{r}$. $(C)$ also can not occur, because in $(C)$, $\iota_{r'}=i$ is the last vertex in $I_{\tau, r}$, forcing the last edge $(\tau_{r}, \tau_{r-1})$ to be a shared edge. 
\item Case $n_{1}=0, n_{2}=2$.  Since both $\iota$ and $\iota_{r'}$ are from shared edges, none of $(A)$, $(B)$, and $(C)$ can occur. This is because in all three cases, the last edge $(\tau_{r}, \tau_{r-1})$  in $I_{\tau, r}$ must be a shared edge, which is not allowed. $(D)$ can not occur either as $\iota_{r'}=\tau$ in the shared edge would contradict $n_{1}=0$.
\item Case $n_{1}=1, n_{2}=2$. We can eliminate $(A)$, $(B)$ and $(C)$ for the same reason as in the $n_{1}=0, n_{2}=2$ case. $(D)$ can not happen either because in $(D)$, $\tau_{r}=\iota$ and $\tau=\iota_{r'}$, and then both $\tau_{r}$ and $\tau$ will be in the shared edges, a contradiction. 
\end{itemize}
\end{remark}

\begin{figure}[ht]
\centering
\includegraphics[width=0.7\textwidth]{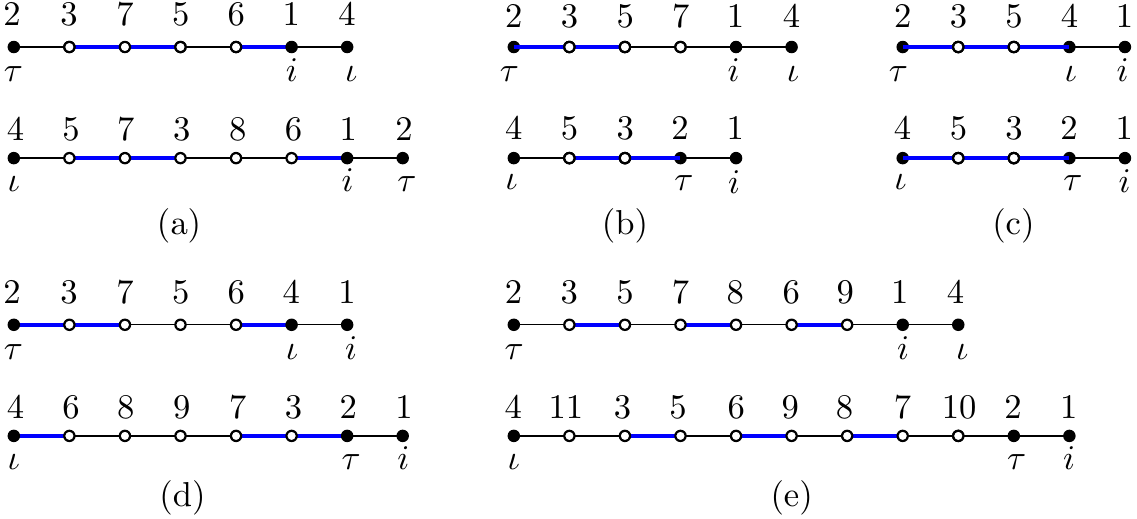}
\caption{\label{fig1} Let $k\geq 9$ and $(i, \tau,\iota)=(1, 2,4)$. These figures are typical examples of elements in $\mathcal{I}_{\tau,\iota,k-5,k-6}(3,5,1)$, $\mathcal{I}_{\tau,\iota,k-4,k-3}(2,3,1)$, $\mathcal{I}_{\tau,\iota,k-3,k-3}(3,4,2),$ $\mathcal{I}_{\tau,\iota,k-5,k-6}(3,5,2)$, $\mathcal{I}_{\tau,\iota,k-7,k-9}(3,6,0)$ from (a) to (e), respectively, where the shared edges are highlighted in blue. To bound the order of the cardinality of $\mathcal{I}_{\tau,\iota,r,r'}(s,t,\ell)$, we only need to consider all possible choices of $\tau_{k-1}, \ldots, \tau_{r+1}$ and $\iota_{k-1}, \ldots, \iota_{r'+1}$ (for example, the open circles in each case) that preserve the self-avoiding property and the number of shared edges. Consequently, from (a) to (e), $|\mathcal{I}_{\tau,\iota,k-5,k-6}(3,5,1)|\leq Cn^{5}$, $|\mathcal{I}_{\tau,\iota,k-4,k-3}(2,3,1)|\leq Cn^{3}$, $|\mathcal{I}_{\tau,\iota,k-3,k-3}(3,4,2)|\leq Cn^{2}$, $|\mathcal{I}_{\tau,\iota,k-5,k-6}(3,5,2)|\leq Cn^{6}$, and $|\mathcal{I}_{\tau,\iota,k-7,k-9}(3,6,0)|\leq Cn^{8}$, where  in each case, the positive constant $C>0$ varies and is independent of $n$.}
\end{figure}

Write
\begin{align}\label{add:eq7}
\mathcal{I}_{\tau,\iota,r,r'}(s,t)=\mathcal{I}_{\tau,\iota,r,r'}(s,t,0)\bigcup \mathcal{I}_{\tau,\iota,r,r'}(s,t,1)\bigcup \mathcal{I}_{\tau,\iota,r,r'}(s,t,2).
\end{align}
The following lemma establishes bounds for the sizes of $\mathcal{I}_{\tau,\iota,r,r'}(s,t,\ell).$

\begin{lemma}\label{lem5}
	For any $1\leq r,r'\leq k-1$, $(s,t)$ satisfying \eqref{add:eq---1},  and $0\leq \ell\leq t$, if $\mathcal{I}_{\tau,\iota,r,r'}(s,t,\ell)$ is nonempty, 
	then 
	\begin{align}\label{lem5:eq1}
	t-\ell\leq \min\bigl(k-r-1,k-r'-1\bigr)
	\end{align}
	 and  there is a constant $C=C(k, r, r', t,\ell)>0$ independent of $n$ such that
	\begin{align}\label{lem5:eq2}
	\bigl|\mathcal{I}_{\tau,\iota,r,r'}(s,t,\ell)\bigr|\leq Cn^{2k-r-r'-t+\ell-2}.
	\end{align}
\end{lemma}

\begin{proof}
 For any $(I_{\tau,r},I_{\iota,r'})\in \mathcal{I}_{\tau,\iota,r,r'}(s,t),$ the first vertices of both paths are already determined and their last edges $(\tau_{r},\tau_{r-1})$ and $(\iota_{r'},\iota_{r'-1})$ are fixed as well due to $(A)-(D)$. Hence, we can only select the vertices, $\tau_{k-1},\ldots,\tau_{r+1}$ and $\iota_{k-1},\ldots,\iota_{r'+1}$, which have cardinalities no larger than $n^{k-r-1}$ and $n^{k'-r'-1}$, respectively. Since there are $t-\ell$ vertices among $\{\tau_{k-1},\ldots,\tau_{r+1}\}$ and $\{\iota_{k-1},\ldots,\iota_{r'+1}\}$ that are shared with each other, \eqref{lem5:eq1} must hold. Also,
	\begin{align*}
	\bigl|\mathcal{I}_{\tau,\iota,r,r'}(s,t,\ell)\bigr|&\leq  Cn^{t-\ell}\cdot n^{(k-r-1)-(t-\ell)}\cdot n^{(k-r'-1)-(t-\ell)}\\
	&=C n^{2k-r-r'-t+\ell-2}
	\end{align*}
	for $0\leq \ell\leq t$, where 
	$$
C=C(k,r,r',t,\ell):=4\cdot (t-\ell)! \binom{k-r-1}{t-\ell} \cdot (t-\ell)!\binom{k-r'-1}{t-\ell}.
$$
	Here, the factor 4 accounts for the four different situations (A)-(D) and the two combinatorial numbers are upper bounds for the numbers of ways that the shared edges in $(I_{\tau,r},I_{\iota,r'})\in \mathcal{I}_{\tau,\iota,r,r'}(s,t,\ell)$ can appear, counting both order and orientation.
\end{proof}
Note that for the unshared edges, the corresponding Gaussian random variables in $A_{I_{\tau,r}}A_{I_{\iota,r'}}$ appear only once and there are $(k-r-s)+(k-r'-s)$ such edges so that we can apply the Gaussian integration by parts to get
\begin{align}\label{add:eq14}
&\e\Bigl[A_{I_{\tau,r}}A_{I_{\iota,r'}}F_{I_{\tau,r}}(A)F_{I_{\iota,r'}}(A)\Bigr]=\e\Bigl[S_{I_{\tau,r},I_{\iota,r'}}\partial_{P_{I_{\tau,r},I_{\iota,r}}}\bigl(F_{I_{\tau,r}}(A)F_{I_{\iota,r'}}(A)\bigr)\Bigr].
\end{align}
Here $S_{I_{\tau,r},I_{\iota,r'}}$ is the product of all $a_{\ell\ell'}$'s with $(\ell,\ell')$ being a shared edge in $(I_{\tau,r},I_{\iota,r'})$ and 
\begin{align}\label{add:eq15}
\e \bigl[S_{I_{\tau,r},I_{\iota,r'}}^2\bigr]\leq \e|z|^{4s}
\end{align}
for $z\thicksim N(0,1).$ The set $P_{I_{\tau,r},I_{\iota,r'}}$ is the collection of unshared edges and $\partial_{P_{I_{\tau,r},I_{\iota,r'}}}$ is the partial derivatives corresponding to the unshared edges in $P_{I_{\tau,r},I_{\iota,r'}}.$ We have the following moment control of these partial derivatives.

\begin{lemma}\label{lem4}
	There exists a constant $C>0$ such that for sufficiently large $n,$ 
	\begin{align*}
	\sup_{(I_{\tau,r},I_{\iota,r'})\in \mathcal{I}_{\tau,\iota,r,r'}(s,t)}\e\bigl|\partial_{P_{I_{\tau,r},I_{\iota,r'}}}\bigl(F_{I_{\tau,r}}(A)F_{I_{\iota,r'}}(A)\bigr)\bigr|^2&\leq \frac{C}{n^{2k-2s-r-r'}}.
	\end{align*}
\end{lemma}

From \eqref{add:eq14}, \eqref{add:eq15}, and Lemma \ref{lem4}, we conclude that there exists some universal constant $C>0$ such that for sufficiently large $n,$
\begin{align}\label{add:eq---2}
\e\Bigl[A_{I_{\tau,r}}A_{I_{\iota,r'}}F_{I_{\tau,r}}(A)F_{I_{\iota,r'}}(A)\Bigr]&\leq  \frac{C}{n^{k-s-(r+r')/2}}.
\end{align}

\begin{proof}[\bf Proof of Lemma \ref{lem4}]
	Recall the terms in the product of \eqref{add:eq10}. For any $m\geq 0$ and $p\geq 1,$ \eqref{lem1:eq2} ensures the existence of positive constants $$W_{k-1,m,p,f_{k-1}'},W_{k-2,m,p,f_{k-2}'},\ldots,W_{r,m,p,f_{r}'},W_{r-1,m,p,f_{r-1}}$$ such that for $n$ large enough, the following inequalities hold,
	\begin{align*}
	\sup_{(P,S,i)\in\mathcal{B}_{s,n}(m)}\Bigl(\e\Bigl|\partial_Pf_{s}'\bigl(w_{S,i}^{[s]}\bigr)\Bigr|^p\Bigr)^{1/p}&\leq \frac{W_{s,m,p,f_s'}}{n^{m/2}},\,\,r\leq s\leq k-1,\\
	\sup_{(P,S,i)\in\mathcal{B}_{r-1,n}(m)} \Bigl(\e\Bigl|\partial_Pf_{r-1}\bigl(w_{S,i}^{[r-1]}\bigr)\Bigr|^p\Bigr)^{1/p}&\leq \frac{W_{r,m,p,f_r}}{n^{m/2}}.
	\end{align*}
	In addition, from Proposition \ref{prop3}, there exists a constant $W_{k,m,p,f_k'}' > 0$ such that
	\begin{align*}
	\sup_{}\Bigl(\e\Bigl|\partial_P\Bigl(f_k'\Bigl(\frac{1}{\sqrt{n}}\sum_{l\neq i,j}a_{jl}f_{k-1}\bigl(w_{\{j\},l}^{[k-1]}(A)\bigr)\Bigr)\Bigr)\Bigr|^p\Bigr)^{1/p}\leq \frac{W_{k,m,p,f_k'}'}{n^{m/2}},
	\end{align*}
	where the supremum is taken over all $P$'s, collections of elements from $\{(i',j'):1\leq i'<j'\leq n\}$ with $|P|=m$ counting multiplicities and $i,j\in [n]$ with $i\neq j.$ These bounds essentially say that each partial derivative will bring up a factor $n^{-1/2}$ module some absolute constant. As a result,  by applying the product rule of the  differentiation, the assertion follows since $|P_{I_{\tau,r},I_{\iota,r'}}|$ is the number of the unshared edges in the pair $(I_{\tau,r},I_{\iota,r'})$ and it is equal to $(k-r-s)+(k-r'-s).$
\end{proof}

Finally, we can bound the off-diagonal term in \eqref{add:eq4} as follows. Using Lemma \ref{lem5} and \eqref{add:eq---2}, we see that for any $1\leq r,r'\leq k-1$, $(s,t)$ satisfying \eqref{add:eq---1}, and $0\leq \ell\leq t$, 
if $\mathcal{I}_{\tau,\iota,r,r'}(s,t,\ell)$ is nonempty, then 
\begin{align*}
&\frac{1}{n^{k+1-(r+r')/2}}\sum_{(I_{\tau,r},I_{\iota,r'})\in \mathcal{I}_{\tau,\iota,r,r'}(s,t,\ell)}\e\Bigl[A_{I_{\tau,r}}A_{I_{\iota,r'}}F_{I_{\tau,r}}(A)F_{I_{\iota,r'}}(A)\Bigr]\\
&\leq \frac{C(k,r, r', t, \ell)}{n^{k+1-(r+r')/2}}\cdot n^{2k-r-r'-t+\ell-2}\cdot\frac{1}{n^{k-s-(r+r')/2}}\\
&=\frac{C(k,r, r', t, \ell)}{n^{3+t-s-\ell}}.
\end{align*}
Here, if $s=0$, then $t=\ell=0$ and  
$$\frac{1}{n^{3+t-s-\ell}} = \frac{1}{n^{3}}.$$  
If $s\ge 1$, using $t\ge s+1$ and $\ell\leq 2,$ we have
$$\frac{1}{n^{3+t-s-\ell}}\le \frac{1}{n^{4-\ell}}\le \frac{1}{n^{2}}.$$ 
As a result, from \eqref{add:eq8}, \eqref{add:eq6}, and \eqref{add:eq7}, for some $C''>0$ independent of  $n$,
\begin{align*}
&\e\bigl[\e_{a_{i\tau}}\bigl[a_{i\tau}L_\iota\bigr]\e_{a_{i\iota}}\bigl[a_{i\iota}L_\tau\bigr]\bigr]\leq  \frac{C''}{n^2}.
\end{align*}
Consequently, this bounds the off-diagonal term in \eqref{add:eq4},
\begin{align}\label{add:eq17}
\frac{1}{n}\sum_{\tau,\iota\neq i:\tau\neq \iota}\e\bigl[\e_{a_{i\tau}}\bigl[a_{i\tau}L_\iota\bigr]\e_{a_{i\iota}}\bigl[a_{i\iota}L_\tau\bigr]\bigr]&\leq \frac{C''}{n}.
\end{align}

\subsection{Step IV: Completion of the proof}

Plugging \eqref{add:eq16} and \eqref{add:eq17} into \eqref{add:eq4} and then using Lemma \ref{lem0}, we see that
\begin{align*}
u_i^{[k+1]}&\asymp_2\frac{1}{\sqrt{n}}\sum_{j\neq i}a_{ij}f_k\Bigl(\frac{1}{\sqrt{n}}\sum_{l\neq i,j}a_{jl}f_{k-1}\bigl(w_{\{j\},l}^{[k-1]}\bigr)\Bigr)\asymp_2 w_i^{[k+1]},\,\,\forall i\in [n].
\end{align*}
This implies that $u^{[k+1]}\asymp_2 w^{[k+1]}$ and completes our proof.

\bibliographystyle{plain}
{\footnotesize\bibliography{ref}}

\end{document}